%% file: RR-8653.tex
\documentclass[twoside]{article}
\usepackage[a4paper]{geometry}
\usepackage{RR}
\usepackage{amssymb,amsmath}
\usepackage{amsthm}
\usepackage{xspace}
\usepackage{enumerate}
\usepackage{graphicx}
\usepackage[utf8]{inputenc}
\usepackage{color}
\usepackage[noend]{algpseudocode}
\usepackage{float}
\usepackage[hypertexnames=false]{hyperref}
\hypersetup{colorlinks=true} %% hyperref en couleurs
\usepackage{algorithm}

\usepackage{breqn}
\newtheorem{Definition}{Definition}%[section]
\newtheorem{Theorem}[Definition]{Theorem}

\newtheorem{Lemma}[Definition]{Lemma}

\newtheorem{Remark}[Definition]{Remark}

\newcommand{\twovector}[2]{\ensuremath{\left(\begin{smallmatrix}#1\\#2\end{smallmatrix}\right)}}
\newcommand{\twomatrix}[4]{\ensuremath{\left(\begin{smallmatrix} #1 & #2\\ #3 & #4\end{smallmatrix}\right)}}
\newcommand{\threevec}[1]{\textbf{#1}}

\newcommand{\ideal}[1]{\langle #1 \rangle }
\newcommand {\CC}   {{\mathcal C}_{P\cap Q}}
\newcommand {\C}   {\mathbb C}
\newcommand {\D}   {\mathbb D}
\newcommand {\R}   {\mathbb R}
\newcommand {\Q}   {\mathbb Q}
\newcommand {\V}   {\mathbb V}
\RRetitle{Numeric certified algorithm for the topology of resultant and discriminant curves}
\RRtitle{Algorithmes numériques certifiés pour la topologie d'une courbe résultante ou discriminante}
\RRauthor{
R\'{e}mi Imbach
\and Guillaume Moroz
\and Marc Pouget
}

\RRequipe{VEGAS}
\RCNancy

\RRabstract{Let $\mathcal C$ be a real plane algebraic curve defined by the resultant of
  two polynomials (resp. by the discriminant of a polynomial). Geometrically
  such a curve is the projection of the intersection of the surfaces
  $P(x,y,z)=Q(x,y,z)=0$ (resp.
  $P(x,y,z)=\frac{\partial P}{\partial z}(x,y,z)=0$), and generically its
  singularities are nodes (resp. nodes and ordinary cusps).  State-of-the-art
  numerical algorithms compute the topology of smooth curves but usually fail to
  certify the topology of singular ones. The main challenge is to find practical
  numerical criteria that guarantee the existence and the uniqueness of a
  singularity inside a given box $B$, while ensuring that $B$ does not contain
  any closed loop of $\mathcal{C}$.  We solve this problem by first providing a
  square deflation system, based on subresultants, that can be used to certify
  numerically whether $B$ contains a unique singularity $p$ or not.  Then we
  introduce a numeric adaptive separation criterion based on interval arithmetic
  to ensure that the topology of $\mathcal C$ in $B$ is homeomorphic to the
  local topology at $p$.  Our algorithms are implemented and experiments show
  their efficiency compared to state-of-the-art symbolic or homotopic methods.}

\RRresume{Bien que francophones et très attachés à notre langue maternelle, nous 
avons pensé et rédigé ce travail en anglais comme la grande majorité de la 
production scientifique mondiale. Dans ce contexte, il est clair que cette 
version française de l' ''abstract'' n'a aucun interêt pour notre communauté, et 
nous avons peu d'espoir qu'il puisse en être autrement même en dehors de notre 
communauté. Nous proposons néanmoins quelques pistes en français pour cet 
improbable lecteur et serions comblés si celui-ci en venait à apprendre 
l'anglais pour pouvoir lire notre prose. Nous \'etudions la topologie d'une 
courbe plane issue de la projection d'une courbe lisse dans l'espace. 
Génériquement, la projection présente des singularités de type noeud et cusp 
(dans le cas d'un discriminant seulement). Les algorithmes numériques de l'état 
de l'art ne calculent la topologie que dans le cas de courbes lisses. L'enjeu 
est donc de concevoir des critères numériques garantissant l'existence et 
l'unicité d'une singularité dans une boite donnée, tout en assurant que cette 
boite ne contienne pas d'autre partie de la courbe non connectée à ce point dans 
la boite. Nous proposons une déflation basée sur les sous-résultants pour le 
premier problème ainsi qu'un critère de séparation basée sur de l'arithmétique 
d'intervalles pour le second problème. } 

\RRkeyword{Topology of algebraic curves, subresultant, numerical algorithm, singularities, interval arithmetic, node and cusp singularities}
\RRmotcle{Topologie de courbes alg\'{e}briques, sous-r\'{e}sultant, algorithme num\'{e}rique, arithm\'{e}tique d'interval, noeuds et cusps}

\begin{document}

\RRNo{8653}

\makeRR
 
% \linenumbers
\section{Introduction}
\label{section-Introduction}
Given a bivariate polynomial $f$ with rational coefficients, a classical problem
is the computation of the topology of the real plane curve
$\mathcal{C}=\{(x,y)\in \R^2 | f(x,y)=0\}$. One may ask for the topology in the
whole plane or restricted to some bounding box. In both cases, the topology is
output as an embedded piecewise-linear graph that has the same topology as the
curve $\mathcal{C}$.  For a smooth curve, the graph is hence a collection of
topological circles or lines; for a singular curve, the graph must report all
the singularities: isolated points and self-intersections.

Symbolic methods based on the cylindrical algebraic decomposition can guarantee
the topology of any curve. % (\cite{mpsttw-ecg-2006} and references therein).
However, the high complexity of these purely algebraic methods prevents them to
be applied in practice on difficult instances. On the other hand, purely
numerical methods such as curve tracking with interval arithmetic or subdivision
are efficient in practice for smooth curves but typically fail to certify the
topology of singular curves. A long-standing challenge is to extend numerical
methods to compute efficiently the topology of singular curves.

Computing the topology of a singular curve can be done in three steps.
\begin{enumerate}
\def\labelenumi{\arabic{enumi}.}
\itemsep1pt\parskip0pt\parsep0pt
\item Enclose the singularities in isolating boxes.
\item Compute the local topology in each box, that is $i)$
compute the number of real branches connected to the singularity, $ii)$ ensure
that it contains no other branches.
\item Compute the graph connecting the boxes.
\end{enumerate}

The third step can be done using existing certified numerical algorithms (e.g.
\cite{Goldsztejn2010,Joris2011,beltran2013robust}), we will thus focus on the
first two steps.

% For the first step, numerical algorithms cannot certify singularities of a curve
% defined by $f$ in the general case.

\paragraph{Contribution and overview.}

  The specificity of the resultant or the discriminant curves computed from
  generic surfaces is that their singularities are stable, this is a classical
  result of singularity theory due to Whitney. The key idea of our work is to
  show that, in this specific case, the over-determined system defining the curve
  singularities can be transformed into a regular well-constrained system of a
  transverse intersection of two curves defined by subresultants. This new
  formulation can be seen as a specific deflation system that does not contain
  spurious solutions.

  %, and as such, encodes spurious
  %solutions.
  %
  %We show that with our deflation system, it is possible to remove these spurious solutions using interval evaluation techniques. In particular our method does not require to evaluate interval up to a separation bound, nor does it require to reduce the problem to a univariate polynomial.
  %
  %To overcome this problem, we show how to discard these spurious
  %solutions by checking inequalities, which is %easily
  %delt with numerically via interval evaluation.

Our contribution focuses on the first two steps of the above mentioned topology
algorithm for a curve defined by the resultant of two trivariate polynomials $P$
and $Q$: $f=Resultant_z(P,Q)$. 

In Section~\ref{sec:singularities}, the main results are
Theorems~\ref{th:sing-sres} and \ref{th:nodecusp} that  characterize the
singularities of the resultant or discriminant curve in terms of subresultants
under generic assumptions. A semi-algorithm~\ref{algo:assuption-check} is
proposed to check these generic assumptions, i.e. it terminates iff the
assumptions are satisfied (note that this is the best one can hope for a purely
numerical method). Based on the characterization of Theorems~\ref{th:sing-sres}
and \ref{th:nodecusp} , Algorithm~\ref{algo:subdiv-sing}, using subdivision and
interval evaluation, isolates the node and cusp singularities with an adaptive
certification.

Sections~\ref{sec:branches}~and~\ref{sec:loopdetection} address the second step
on the above mentioned topology algorithm, that is computing the local topology
at singularities.
%The general idea is to adaptively refine the isolating boxes
%of the singular points until some inequalities are satisfied. 
Algorithms~\ref{algo:nb-branches} and \ref{algo:nb-branches-discrim} in
Section~\ref{sec:branches} distinguish nodes from cusps and compute the number
of branches.  Then in Section~\ref{sec:loopdetection},
Algorithm~\ref{algo:avoid-loop} certifies that an isolation box of a singular
point does not contain locally other branches than those that pass through the
singularity. 

In Section~\ref{sec:experiments}, experiments are detailed showing that our
specialized certified numerical method outperforms state-of-the-art
implemented methods for polynomials of degree greater or equal to $5$. Moreover, the performance of our method is also improved when we restrict the problem to a box.
%both symbolic and homotopy methods.

% Our contribution focuses on the first two steps of the above mentioned topology
% algorithm for a curve defined by the resultant of two trivariate polynomials $P$
% and $Q$: $f=Resultant_z(P,Q)$.  We show that it is possible to certify its
% singularities and compute its local topology with adaptive numerical algorithms.
% % However, when $f$ is the resultant of two trivariate polynomials $P$ and $Q$, we
% % show that it is possible to certify its singularities and compute its local
% % topology with adaptive numerical algorithms.
% %Let $f=Resultant_z(P,Q)$. 
% More precisely, Section \ref{sec:singularities} presents a numerical algorithm
% to isolate and certify nodes and ordinary cusp points of $f$. Section
% \ref{sec:resultant} describes how to compute the local topology of nodes of
% $f$. Finally, when $Q=\frac{\partial P}{\partial z}$, Section
% \ref{sec:discriminant} describes how to compute the local topology of nodes and
% ordinary cusps of $f$.

\paragraph{Notations.}

Let $f$ be a bivariate polynomial and $\mathcal{C}$ it associated curve. We
denote by $f_{x^iy^j}$ the partial derivative $\frac{\partial^{i+j} f}{(\partial
x)^i(\partial y)^j}$.  A point $p=(\alpha,\beta)$ in $\C^2$ is \emph{singular}
for $f$ if $f( { p}) = f_x( { p}) = f_y( { p}) = 0$, and \emph{regular}
otherwise. A \emph{node} is a singular point with
$\det(\text{Hessian}(f))=f_{xy}^2-f_{x^2}f_{y^2}\neq0$. An \emph{ordinary cusp}
is a singular point such that $\det(\text{Hessian}(f))=0$ and for all non
trivial direction $(u,v)$, $f(\alpha +ut,\beta+vt)$ vanishes at $t=0$ with
multiplicity at most $3$.

We denote by $\Box f$ any convergent interval extension of $f$, that is for any
box $B$, $\{f(x,y) | (x,y)\in B\} \subset \Box f(B)$, and for any decreasing
sequence of boxes $B_i$ converging to a point $p$, the sequence $\Box f(B_i)$
converges to $f(p)$.
% Given an interval box $B\subset\R^2$, we denote by $\Box f(B)$ the interval of
% $\R$ defined by replacing the arithmetic operations $+,-,\times$ by the
% corresponding interval operations \MP{il semble que l'on peut dire qu'il faut
%   slt que ce soit convergent, inclusion monotone aussi?} (note that due to the
% non-distributivity of interval operations, different evaluation schemes of the
% same function $f$ give different interval functions). 
By abuse of notation, we often simply denote $\Box f(B)$ by $\Box f$. The
Krawczyk operator of a mapping $F$ defined in Lemma~\ref{cri:regsol} is denoted
by $K_F$.

For two polynomials $P$ and $Q$ in $\D[z]$ with $\D$ a unique factorization
domain (in this article $\D$ will be $\Q[x,y]$), recall that the $i^{\text{th}}$
subresultant polynomial is of degree at most $i$ (see e.g. \cite[\S
3]{Kahoui03}), we denote it $S_i(z)=s_{ii}z^i + s_{i,i-1}z^{i-1} + \dots +
s_{i0}$. The resultant is thus $S_0(z)=s_{00}$ in $\D$ and we also denote it
more classically as $Res_z(P,Q)$. Finally, $\V(f_1,\ldots,f_n)$ denotes the
solutions of the system $f_1=\cdots=f_n=0$.

% The variety of an ideal $I$ is noted $\V(I)$ and its Zariski closure
% $\overline{\V(I)}$. 
% \MP{check}

\paragraph{Previous and related work.}
%%%%%%%%%%%%%%%%%%%%%
\label{sec:previous}

There are many works addressing the topology computation via symbolic methods,
see for instance the book chapter \cite{mpsttw-ecg-2006} and references
therein. Most of them use subresultant theory, but there are also some
alternatives using only resultants
(e.g. \cite{WolSei:topology:05,sagraloff11-bisolve-complexity}) or Gröbner bases
and rational univariate representations \cite{isotop-mcs-10}. Some alternative even compute a rational univariate representation numerically if all approximate solutions are known \cite{AkogluHS14}.
For the restricted case of computing the topology of non-singular curves, numerical methods are usually
faster and can in addition reduce the computation to a user defined bounding
box. One can mention 
interval analysis methods \cite{Goldsztejn2010} or more generally certified homotopy methods \cite{beltran2013robust,Joris2011}. %, both being based on variants of the Newton operator \cite{Rump1983}.
These methods are based on the fact that the regular solutions of a square
system can be certified and approximated with quadratic convergence with the
interval Newton-Krawczyk operator \cite{Rump1983,Nbook90}.
Another well-studied numerical approach is via recursive subdivision of the
plane. Indeed, the initial idea of the marching cube algorithm
\cite{Lorensen:1987:MCH:37402.37422} can be further improved with interval
arithmetic to certify the topology of smooth curves
\cite{Snyder1992,plantinga04,liang08}. %,burr-csaII-08}.alberti-topo-cagd-08?

For singular curves, isolating the singular points is already a challenge from a
numerical point of view. Indeed, singular points are defined by an
over-determined system $f=f_x=f_y=0$ and are not necessarily regular solutions
of this system. A classical approach to handle an over-determined system
$\{f_1, \ldots, f_m\}$ is to combine its equations in the form
${f_1}_{x_i}f_1+\cdots+{f_m}_{x_i}f_m=0$ for each variable
$\{x_i\}_{ 1\leq i \leq n < m}$,
%or to consider several random subsystems, 
to transform it into a square system
\cite{dedieu2006points},
% \MP{dedieu ne fait pas ca mas utilise les moindre carre
%   et on peut alors dire: use a least square formulation and hence transform the
%   problem into a global optimization; sinon trouver une ref qui fait des
%   combinaisons}\GM{en fait je considère que les moindres carrés sont une
%   combinaison avec en facteur les dérivés des polynomes mais c'est
%   effectivement pas directe, j'ai modifé en explicitant les équations}
but this introduces spurious solutions. Singular solutions can be handled
through deflation \cite{Giusti2007,Ojika1983,Leykin2006,Mantzaflaris2011Issac},
roughly speaking, the idea is to compute partially the local structure of a
non-regular solution, and use this information to create a new system where this
solution is regular.  However this system is usually still overdetermined, and
it does not vanish on the solutions of the original system that do not have the
same local structure.  Thus, this cannot be directly used to separate solutions
with different multiplicity structures. It is important to mention that
  the certification of solutions of over-determined systems is theoretically out
  of reach of numerical methods in the general case. In the polynomial case,
  non-adaptive lower bounds can be used but they are too pessimistic to be
  practical, see \cite[Remark 7]{Hauenstein2012} or \cite{Burr2012}.
%A recent
%  symbolic-numeric alternative approach is proposed in \cite{AkogluHS14} via
%  computations of exact rational univariate representations from approximate
%  solutions.
  
When the curve we consider is a resultant, its singular locus can be
related to the first subresultant
(see \cite[\S 4.3]{1979Jouanolou} and \cite[\S 5]{buse2009} for examples). 
In Section \ref{sec:singularities}, we use this structure to exhibit a square
deflation system. Another approach would be to exhibit a square system in higher
dimension that defines the set of points for which the polynomials $P$ and $Q$
have two solutions. This approach was considered in \cite{2013-delanoue-jcam} to
compute the topology of the apparent contour of a smooth mapping from
$\mathbb{R}^2$ to $\mathbb{R}^2$.

The number of real branches connected to the singularity can be computed with
the topological degree of a suitable mapping
\cite{Szafraniec88,alberti-topo-cagd-08,Mantzaflaris2011Issac} or with the fiber
multiplicity together with isolation on the box boundary
\cite{WolSei:topology:05}.
Certifying the topology inside a box requires the detection of loops near a
singularity. It is usually solved in the literature by isolating the
$x$-extreme points, which reduces the problem to a univariate polynomial
computed with resultants (\cite{WolSei:topology:05,mpsttw-ecg-2006} for
example).

We are not aware of numerical algorithms that can handle in practice the
computation of the topology of singular curves.
Still, relying on global non-adaptative separation bounds for algebraic
systems, the subdivision approach presented in \cite{Burr2012} can theoretically
certify the topology of any singular curve. Due to these worst-case bounds, this
algorithm cannot be practical. 
A numerical algebraic geometric approach is presented in \cite{lu2007} using
irreducible decomposition, generic projection and plane sweep, deflation and
homotopy to compute the topology of a singular curve in any codimension. So far
this work seems more theoretical than practical and the certification of all the
algorithm steps appears as a challenge.

\section{Subresultant based deflation}
%%%%%%%%%%%%%%%%%%%%%%%%%%%%%%%%%%%%
\label{sec:singularities}
The input of algorithms in this section are two trivariate polynomials $P, Q$ 
%in $\Q[x,y,z]$ 
and a box $B_0$ in $\R^2$.  Our goal is to isolate the singularities of the
plane curve $f=0$ defined by the resultant of $P$ and $Q$ with respect to
$z$. 
% The main result is Theorem~\ref{th:sing-sres} exhibiting a square
% polynomial system $g=h=0$ and two polynomials $u$ and $v$ such that the
% singularities of $f$ are exactly the solutions of the system $g=h=u=0$ and of
% the system $g=h=0$ and $v\neq 0$.  
In this section, we exhibit a square polynomial system $g=h=0$ and a polynomial
$u$ such that the singularities of $f$ are exactly the solutions of the
constrained system $g=h=0$ and $u\neq 0$. Moreover, the singularities are
regular solutions of $g=h=0$, such that numerical methods can certify whether a
box contains or not a singularity. In Section~\ref{sec:sing-sres}, the
constrained system is constructed using subresultants.  In
Section~\ref{sec:regularity}, the regularity of this system is translated in
terms of types of singularities.  Generic assumptions are required so that these
characterizations of the singularities of $f$ hold.
Section~\ref{sec:checkgeneric} presents a semi-algorithm for checking the
assumptions that we now define.
Given two trivariate polynomials $P, Q$ in $\Q[x,y,z]$ and a two-dimensional box
$B_0$, we define the generic assumptions:
\begin{itemize}%\cramped
\item[$(A_1)$] Above the box $B_0$ for the $x$ and $y$-coordinates, 
  % In the box $B_0\times \R$, \MP{or $B_0\times \C$ ??} 
  the intersection of the surfaces $P(x,y,z)=0$ and $Q(x,y,z)=0$ is a smooth
  space curve denoted $\CC$, i.e. the tangent vector $\threevec{t}=
  \triangledown P \times \triangledown Q$ is nowhere null on $\CC$ (where
  $\triangledown P$ is the gradient vector $(P_x,P_y,P_z)$).
\item[$(A_2)$] Above any point $(\alpha,\beta)$ in $B_0$, there are at most two
  points of $\CC$ counted with multiplicities, or in other words, the polynomial
  $\gcd(P(\alpha,\beta,z),Q(\alpha,\beta,z))$ has degree at most two. In
  addition, there are finitely many $(\alpha,\beta)$ in $B_0$ such that this
  degree is two.
\item[$(A_3)$] The leading coefficients $L_P(x,y)$ and $L_Q(x,y)$ of $P$ and $Q$
  seen as polynomials in $z$ have no common solutions in $B_0$.
\item[$(A_4)$] The singularities of the resultant or discriminant curve are only nodes or ordinary cusps.
%regular solutions of the system $\{s_{11}, s_{10}\}$.
  % avoid points or curve at infinity
  % \item only nodes: when 2 points above, the tangents are not in the same
  %   vertical plane.  only used to identify cusps with detH=0
\end{itemize}

\noindent
Note that these assumptions are satisfied for almost all pairs of polynomials in $\Q[x,y,z]$.

% Note that the generic condition $(A_3)$ enables to specialize the subresultant
% sequence \cite[Lemma 3.1]{Kahoui03}.
% , that is the evaluation of the subresultant $S_i(x,y,z)$ at
% $(x,y)=(\alpha,\beta)$ equals the subresultant computed from the polynomial
% $P(\alpha,\beta,z)$ and $Q(\alpha,\beta,z)$ up to a scalar factor.

\subsection{Singularities via subresultants}
%%%%%%%%%%%%%%%%%%%%%%%%%%%%%%%%%%%%
\label{sec:sing-sres}

Let $f$ be the resultant polynomial (with respect to the variable $z$) of two
polynomials $P$ and $Q$ in $\Q[x,y,z]$. We always assume that $f$ is square-free
and thus its singularities are isolated.  
Let $S_{sing}=\V(f,f_xf_y)$ be the set of singular points of $f$ and
$S_{sres}=\V(s_{11},s_{10})-\V(s_{22})$. We prove in this section that, under
our assumptions, these two sets coincide.

\begin{Theorem}[\cite{recknagel13}]
  \label{th:sing-sres}
  Let $f$ be the resultant of the polynomials $P$ and $Q$ in $\Q[x,y,z]$ with
  respect to the variable $z$. Then $S_{sres} \subset S_{sing}$ and if the
  assumptions $(A_1)$ to $(A_3)$ are satisfied then $S_{sing} \subset S_{sres}$.
  % \MP{or if the assumptions $(A_1)$ to $(A_3)$ are satisfied on box $B_0$ then
  %   $S_{sing}\cap B_0 \subset S_{sres}\cap B_0$.}
 % \begin{itemize}
 %  \item $S_{sres} \subset S_{sing}$,
 %  \item if the generic conditions $(A_1)$ to $(A_3)$ are satisfied then
 %    $S_{sing} \subset S_{sres}$.
 %  \end{itemize}
\end{Theorem}

\paragraph{Proof of the inclusion $S_{sres} \subset S_{sing}$.}
%%%%%%%%%%%%
Let $I=\ideal{f,f_x,f_y}$ and
$J=\ideal{s_{11},s_{10}}:\ideal{s_{22}}^\infty$, then $S_{sing}=\V(I)$ and
$\V(J)=\overline{\V(s_{11},s_{10})-\V(s_{22})}=\overline{S_{sres}}\supset
S_{sres}$. It is thus sufficient to prove that $I\subset J$, or in other words
that there exists a positive integer $m$ such that $\ideal{f,f_x,f_y}\cdot
\ideal{s_{22}}^m = \ideal{s_{22}^mf,s_{22}^mf_x,s_{22}^mf_y}\subset
\ideal{s_{11},s_{10}}$.

The generic chain rule of subresultant (see for instance \cite[Theorem
4.1]{Kahoui03}) yields $s_{22}^2 f = Res(S_2,S_1)$. On the other hand, 
$Res(S_2,S_1)=
\left| \begin{matrix}
  s_{22} & s_{11} & \\
s_{21} & s_{10} & s_{11}\\
s_{20} & & s_{10}
\end{matrix}
\right|
= s_{10}^2s_{22}+s_{11}^2s_{20}-s_{10}s_{11}s_{21}.
$ 
Hence $s_{22}^2 f \in \ideal{s_{11},s_{10}}$.

The previous identity expresses $s_{22}^2f$ as a quadratic form in $s_{11}$
and $s_{10}$, differentiating with respect to $x$ (or $y$) yields a sum with
$s_{11}$ or $s_{10}$ as a factor in each term, thus $\partial (s_{22}^2f)$ is
in $\ideal{s_{11},s_{10}}$. This implies that $\partial (s_{22}^3f)$ is also
in $\ideal{s_{11},s_{10}}$. In addition, $\partial (s_{22}^3f)=
3s_{22}^2f\partial s_{22} + s_{22}^3\partial f$ hence $s_{22}^3\partial f
= \partial (s_{22}^3f)- 3s_{22}^2f\partial s_{22}$ with both terms in
$\ideal{s_{11},s_{10}}$, thus $\partial (s_{22}^3f)$ is in
$\ideal{s_{11},s_{10}}$.  We conclude that
$\ideal{s_{22}^3f,s_{22}^3f_x,s_{22}^3f_y}\subset \ideal{s_{11},s_{10}}$,
hence $I\subset J$ and $S_{sres} \subset S_{sing}$.

\paragraph{Proof of the inclusion $S_{sing} \subset S_{sres}$.}
%%%%%%%%%%%%
% Generic condition $(A_3)$ enables to specialize the subresultant sequence
% \cite[Lemma 3.1]{Kahoui03}), that is the evaluation of the subresultant
% $S_i(x,y,z)$ at $(x,y)=(\alpha,\beta)$ equals the subresultant computed from the
% polynomial $P(\alpha,\beta,z)$ and $Q(\alpha,\beta,z)$ up to a non-null scalar
% factor in $\Q(\alpha,\beta)$ \MP{use footnote for the weird cse}. Thus a
% subresultant coefficient $s_{i,j}$ vanishes when evaluated at $(\alpha,\beta)$
% iff the corresponding coefficient of the subresultant sequence of
% $P(\alpha,\beta,z)$ and $Q(\alpha,\beta,z)$ is null.
% %

Let $(\alpha,\beta)$ be a singular point of $f$, so that $f(\alpha,\beta)=0$.
According to the generic condition $(A_2)$,
$\gcd(P(\alpha,\beta,z),Q(\alpha,\beta,z))$ has at most two simple roots or one
double root.

For the case of a double root, $\gcd(P(\alpha,\beta,z),Q(\alpha,\beta,z))$ has
degree 2 and by the gap structure theorem (more precisely its corollary showing
the link between the gcd and the last non-vanishing subresultant, see
e.g. \cite[Corollary 5.1]{Kahoui03}) and assumption $(A_3)$: (a) this gcd is the
subresultant $S_2(\alpha,\beta)$, hence $s_{22}(\alpha,\beta)\neq0$, and (b)
the subresultants of lower indices are vanishing, in particular
$s_{11}(\alpha,\beta)=0$ and $s_{10}(\alpha,\beta)=0$. Hence $(\alpha,\beta)$
is in $S_{sres}$.

Otherwise, let $\gamma$ be a simple root of
$\gcd(P(\alpha,\beta,z),Q(\alpha,\beta,z))$, the generic condition $(A_1)$
yields that the tangent vector $\threevec{t}(p)$ to $\CC$ at the point
$p=(\alpha,\beta,\gamma)$ is well defined and not vertical. Indeed, the
multiplicity of $\gamma$ in $\gcd(P(\alpha,\beta,z),Q(\alpha,\beta,z))$ is 1, so
it is also one in at least one of the polynomials $P(\alpha,\beta,z)$ or
$Q(\alpha,\beta,z)$. In other words, $P_z(p)\neq0$ or $Q_z(p)\neq0$ which implies that the $x$ and $y$-coordinates of
$\threevec{t}(p)$ cannot both vanish (otherwise, $\threevec{t}(p)$
would be the null vector contradicting assumption $(A_1)$).  Without loss of
generality we may assume that the $x$-coordinate of $\threevec{t}(p)$ is
not null: $x_{\threevec{t}(p)}=P_y(p)Q_z(p)-P_z(p) Q_y(p)\neq 0$.

We now apply \cite[Theorem 5.1]{buse2009} rephrased in the affine setting to $P$
and $Q$:
% We now rephrase an affine formulation of a lemma originally stated in the
% homogeneous case.
% \begin{Lemma}[{\cite[Theorem 5.1]{buse2009}}]
% \label{lem:buse5.1}
%   Let $f_1$ and $f_2$ be two polynomials of degree $d_1$ in $y$ and $d_2$ in $z$
%   with coefficients in an integral domain. Then 
% $$ \partial_y Res_z(f_1,f_2)= (-1)^{d_1+d_2}
% \left| \begin{matrix}
%  \partial_y f_1 & \partial_z f_1\\
%  \partial_y f_2 & \partial_z f_2\\
% \end{matrix}
% \right|
% s_{11} 
% \quad \bmod \ideal{f_1,f_2}.
% $$
% \end{Lemma}
%
%This lemma applied to $P$ and $Q$ yields
$$f_y= \pm
\left| \begin{matrix}
 P_y &  P_z\\
 Q_y & Q_z\\
\end{matrix}
\right|
s_{11} 
+ uP+vQ
$$
with $u,v$ in $\Q[x,y]$. Evaluated at $p$, $P$ and $Q$ vanish and we obtain:
$f_y(\alpha,\beta)=\pm x_{\threevec{t}(p)}
s_{11}(\alpha,\beta)$. Since $(\alpha,\beta)$ is a singular point of $f$,
$f_y(\alpha,\beta)=0$, and together with
$x_{\threevec{t}(p)}\neq0$ this gives $s_{11}(\alpha,\beta)=0$.  The gap
structure theorem and $f(\alpha,\beta)=0$ then implies that (a)
$s_{10}(\alpha,\beta)=0$, and (b) the degree of
$\gcd(P(\alpha,\beta,z),Q(\alpha,\beta,z))$ is at least two. Together with the
generic condition $(A_2)$, this degree is exactly two and so is the degree of
the second subresultant $S_2$ evaluated at $(\alpha,\beta)$, thus
$s_{22}(\alpha,\beta)\neq0$. We then conclude that in this case too
$(\alpha,\beta)$ is in $S_{sres}$.

\subsection{Regularity conditions}
\label{sec:regularity}
The main theorem of this section is the relation between the types
of singularities  of $f$ and the regularity of the solutions of the system
$s_{11}=s_{10}=0$. We assume for this section that the assumptions \((A_1)\),
\((A_2)\) and \((A_3)\) hold.

\begin{Theorem}
\label{th:nodecusp}
Let \(f\) be the resultant of the polynomials \(P\) and \(Q\) in
\(Q[x,y,z]\) with respect to the variable \(z\). If the assumptions
\((A_1)\), \((A_2)\) and \((A_3)\) hold then the following propositions
are equivalent:
\begin{enumerate}
\def\labelenumi{\roman{enumi}.}
\itemsep1pt\parskip0pt\parsep0pt
\item
  \(p\) is a regular solution of \(s_{11}=s_{10}=0\) and
  \(s_{22}(p)\neq  0\)
\item
  \(p\) is a node or an ordinary cusp of the curve \(f=0\)
\end{enumerate}
Furthermore in this case, $p$ is an ordinary cusp point if and only if $\CC$ has a vertical tangent above $p$.
\end{Theorem}

%Before proving this theorem, we will need the following Lemma.
The proof of this theorem is decomposed with the following lemmas.

\begin{Lemma}[\cite{recknagel13}]
  \label{lem:node}
  Let $p$ be a node of $f$. Then $p$ is
  a regular point of the system $s_{11}=s_{10}=0$.
\end{Lemma}

\begin{proof}
  Since $p$ is a node, it is a singular point of $f=0$ and
  Theorem~\ref{th:sing-sres} implies that $p$ is a solution of the system
  $s_{11}=s_{10}=0$. Moreover, we saw in the proof of Theorem~\ref{th:sing-sres}
  that $S_{sres}\subset S_{sing}$ but more precisely that
  $\ideal{s_{22}^3f,s_{22}^3f_x,s_{22}^3f_y}\subset \ideal{s_{11},s_{10}}$.  In
  particular, this implies that the multiplicity of $p$ in
  $\ideal{s_{11},s_{10}}$ is lower or equal to its multiplicity in
  $\ideal{s_{22}^3f,s_{22}^3f_x,s_{22}^3f_y}$. Since $p$ is a node of $f$, the
  determinant of the Hessian of $f$ is non-zero and $p$ is a regular point of
  $\ideal{f,f_x,f_y}$. And since $s_{22}(p)\neq 0$, we can conclude that the
  multiplicity of $p$ in $\ideal{s_{22}^3f,s_{22}^3f_x,s_{22}^3f_y}$ is $1$.
  Thus $p$ has also a multiplicity one in $\ideal{s_{11},s_{10}}$.
\end{proof}

\begin{Lemma}
  \label{lem:cusp}
  Let $p$ be an ordinary cusp point of $f$. Then $p$ is a regular point of the system $s_{11}=s_{10}=0$.
\end{Lemma}

\begin{proof}
  Let $p=(\alpha, \beta)$ be an ordinary cusp point of $f$. Suppose by
  contradiction that $p$ is a singular solution of $s_{11}=s_{10}=0$. Then the
  determinant of the Jacobian matrix
  $\twomatrix{{s_{11}}_x}{{s_{10}}_x}{{s_{11}}_y}{{s_{10}}_y}$ is $0$ and there
  exists a vector $(u,v)\in \mathbb{R}^2\setminus\{(0,0)\}$ orthogonal
  simultaneously to the gradient of $s_{11}$ and to the gradient of
  $s_{10}$. In particular,
  $s_{11}(\alpha+ut,\beta+vt)$ (resp. $s_{10}(\alpha+ut,\beta+vt)$)
  vanishes at $0$ in $t$ with multiplicity at least $2$. 
% In addition, by
%   assumptions, $p$ is an ordinary cusp of $f$ and the polynomial
%   $f(\alpha+ut, \beta+vt)$ vanishes at $0$ in $t$ with multiplicity at
%   most $3$.
  %
  Using standard formula on the resultants (\cite[Theorem 4.1]{Kahoui03}
  for example) we have $s_{22}^2f = Res(S_1, S_2)$. Developing the right
  hand side we get:
  $$s_{22}^2f = s_{22}s_{10}^2 - s_{21}s_{11}s_{10} + s_{20}s_{11}^2.$$
  Thus, evaluating the right hand side on $(\alpha+ut,\beta+vt)$, we
  observe that it vanishes at $0$ in $t$ with multiplicity at least $4$.
  
  On the other hand, $p$ being an ordinary cusp of $f$, the polynomial
  $f(\alpha+ut, \beta+vt)$ vanishes at $0$ in $t$ with multiplicity at most
  $3$. In addition, under the assumptions $(A_2)$ and $(A_3)$, we have
  $s_{22}(p)\neq 0$ and the left hand side vanishes at $0$ in $t$ with
  multiplicity at most $3$, hence the contradiction.
\end{proof}

\begin{Lemma}
\label{lem:nosing}
Let $q=(\alpha,\beta,\gamma)$ be a regular point of the curve $\CC$ such that
$s_{22}(p)\neq 0$ with $p=(\alpha,\beta)$. Then $q$ is a regular point of the
curve \(S_2(x,y,z)=S_1(x,y,z)=0\).
Moreover, the vectors $\nabla P(q),\nabla Q(q)$ generate the same vector space
as $\nabla S_2(q)$ and $\nabla S_1(q)$.
% Under the assumptions \((A_1)\),
% \((A_2)\)
% and \((A_3)\),
% the points of the curve \(P(x,y,z)=Q(x,y,z)=0\)
% are regular points of \(S_2(x,y,z)=S_1(x,y,z)=0\).
% Moreover for all $q \in \CC$, the vectors $\nabla P(q),\nabla Q(q)$ generate
% the same vector space as $\nabla S_2(q)$ and $\nabla S_1(q)$.
\end{Lemma}

\begin{proof}
Using the identities of \cite[Theorem 4.2]{Kahoui03}, there exists
\(U,V,U',V'\) such that:
\begin{align*}
    s_{22}^2 P & = U S_2 + V S_1\\
    s_{22}^2 Q & = U'S_2 + V'S_1\\
\end{align*}
Since $s_{22}(p)\neq 0$, we have:
\begin{align*}
    \nabla P(q) & = \frac{U(q)}{s_{22}(p)^2} \nabla S_2(q)  + \frac{V(q)}{s_{22}(p)^2} \nabla S_1(q) &
    \nabla Q(q) & = \frac{U'(q)}{s_{22}(p)^2} \nabla S_2(q) + \frac{V'(q)}{s_{22}(p)^2} \nabla S_1(q)
\end{align*}
Since \(q\) is a regular point of \(P=Q=0\), 
\(\nabla P(q)\) and \(\nabla Q(q)\) generate a dimension \(2\) vector space.
Thus \(\nabla S_2(q)\) and \(\nabla S_1(q)\) also generate the same dimension \(2\)
vector space and \(q\) is a regular point of the curve \(S_2=S_1=0\).
\end{proof}

\begin{proof}[Proof of Theorem \ref{th:nodecusp}]
  The implication ii. \(\implies\) i. is a direct corollary of Lemma \ref{lem:node} and
  \ref{lem:cusp}.

For the reciprocal, we know that $s_{22}(p)\neq 0$, thus
\[ f = \frac1{s_{22}}s_{10}^2 + \frac1{s_{22}^2}s_{20}s_{11}^2
       - \frac1{s_{22}^2}s_{21}s_{10}s_{11} \]
Let us denote by \(A\), \(J\) and \(V\) the matrices and the
vector
\begin{align*}
  A & = \left(\begin{array}{cc} 2 s_{22} & -s_{21} \\
                                -s_{21} & 2 s_{20} \end{array}\right) &
  J & =  \left(\begin{array}{cc} {s_{10}}_x & {s_{10}}_y \\
                                 {s_{11}}_x & {s_{10}}_y \end{array}\right)&
  V & = \left(\begin{array}{c} s_{10} \\ s_{11} \end{array}\right) &
\end{align*}
The resultant satisfies \(f = \frac1{2 s_{22}^2} V^t \cdot A \cdot V\).
Let \(p\)
be a singular point of the curve \(f=0\).
According to Theorem~\ref{th:sing-sres}, \(s_{11}(p)=s_{10}(p)=0\).
Moreover, without restriction of generality, we can assume that $(\alpha,\beta,0)$
satisfy $P(\alpha,\beta,0)=Q(\alpha,\beta,0)=0$ using the property that the resultant is
invariant by translation of \(z\)
in \(P\) and \(Q\). In this case, we have also \(s_{20}(p)=0\).

With abuse of notations, we denote by \(O_k(x,y)\) a polynomial that is
in the ideal \(\ideal{x,y}^k\) where \(k\) is a positive integer. In
particular we have:
\begin{align*}
   O_{k_1}(x,y) \cdot O_{k_2}(x,y) & = O_{k_1+k_2}(x,y) \\
   O_{k_1}(x,y) +     O_{k_2}(x,y) & = O_{min(k_1,k_2)}(x,y) \\
   \delta O_{k}(x,y) & = O_{k-1}(x,y) \text{ for } \delta=\partial/\partial_x
                       \text{ or } \partial/\partial_y
\end{align*}
With this notation, the taylor expansion of $V$ at $p$ gives

\[V(p+ (x,y))=  J(p) \left(\begin{array}{c}
x\\y\end{array}\right)+ O_2(x,y)  \]
such that :
\[f(p+ (x,y))= \frac1{2 s_{22}(p)^2}\left(x\quad y\right)J(p)^t \cdot A(p) \cdot J(p) \left(\begin{array}{c}
x\\y\end{array}\right)+ O_3(x,y)\]
This implies that the Hessian of \(f\) at $p$ is the matrix
\(\frac1{s_{22}(p)} J(p)^t \cdot A(p) \cdot J(p)\). If the determinant of the
Hessian is not zero, then \(p\) is a node. Otherwise we have
\(\det(A(p))\det(J(p))^2=0\). Let us prove in this case that \(p\) is
an ordinary cusp in \(f\). For that, we need to prove that for every
direction \((u,v)\neq(0,0)\), the valuation of \(t\) in \(f(ut,vt)\) is
lower or equal to \(3\). By hypothesis $i.$, \(\det(J)\neq 0\), thus
\(\det(A(p)) = 4s_{22}(p)s_{20}(p)-s_{21}(p)^2 = 0\). In particular, this
means that \(s_{21}(p)=0\). In particular recalling that:
\[ f = \frac1{s_{22}}s_{10}^2 + \frac1{s_{22}^2}s_{20}s_{11}^2
       - \frac1{s_{22}^2}s_{21}s_{10}s_{11} \]
we have for \((u,v)\) such that
\(a := u{s_{10}}_x(p)+v{s_{10}}_y(p)\neq0\):

\begin{align*}
s_{10}^2          (\alpha+ut,\beta+vt) &= a^2t^2 + O_3(t)\\
s_{20}s_{11}^2    (\alpha+ut,\beta+vt) &= O_3(t)\\
s_{21}s_{10}s_{11}(\alpha+ut,\beta+vt) &= O_3(t)
\end{align*}
This implies:
\[ f(\alpha+ut,\beta+vt) = \frac1{s_{22}(p)^2}a^2t^2 + O_3(t) \]
and for \((u,v)\) such that \(u{s_{10}}_x(p)+v{s_{10}}_y(p)=0\) there
exists a constant \(c\neq0\) such that
\((u,v) = (c{s_{10}}_y, -c{s_{10}}_x)\) and we have:

% \nolinenumbers
% \begin{nolinenumbers}
\begin{dmath*}
f(\alpha+ut,\beta+vt) = \frac{c^3}{s_{22}(p)^2}({s_{20}}_x(p){s_{10}}_y(p)-{s_{20}}_y(p){s_{10}}_x(p))
                                  ({s_{11}}_x(p){s_{10}}_y(p)-{s_{11}}_y(p){s_{10}}_x(p))^2t^3
                                  + O_4(t)
         = \frac{c^3}{s_{22}(p)^2}\det(G(p))\det(J(p))^2t^3 + O_4(t)
\end{dmath*}
% \end{nolinenumbers}
where 
\[ G := \left(\begin{array}{cc} {s_{20}}_x & {s_{20}}_y\\
                               {s_{10}}_x & {s_{10}}_y \end{array}\right)\]
Lemma \ref{lem:nosing} implies that $(\alpha,\beta,0)$ is a regular point of
$S_2(x,y,z)=S_1(x,y,z)=0$. On the other hand,

\begin{align*}
\nabla S_1(\alpha,\beta,0) = \left( {s_{10}}_x(p) \quad {s_{10}}_y(p) \quad s_{11}(p) \right)\\
\nabla S_2(\alpha,\beta,0) = \left( {s_{20}}_x(p) \quad {s_{20}}_y(p) \quad s_{21}(p) \right)
\end{align*}
Since \(s_{11}(p)=s_{21}(p)=0\), the point $(\alpha,\beta,0)$ is regular in
$S_2(x,y,z)=S_1(x,y,z)=0$ only if the determinant of the matrix \(G(p)\) is
different from zero. In addition, hypothesis $i.$ implies
\(\det(J(p))\neq0\). We thus conclude that for 
every \((u,v)\neq(0,0)\), the valuation of \(t\) in \(f(ut,vt)\) is
lower or equal to \(3\), and \(p\) is an ordinary cusp.

Finally, we prove that $p$ is an ordinary cusp if and only if $\CC$ has a
vertical tangent above $p$ at $q=(\alpha,\beta,0)$. First, if $p$ is an ordinary
cusp, then the Hessian of $f$ is zero at $p$ and $\det(A(p))=0$. In this case we
saw that $s_{21}(p)=0$ and since $s_{11}(p)=0$, this implies that
$\frac{\partial S_2}{\partial z}(q) = s_{21}(p)=0$ and
$\frac{\partial S_1}{\partial z}(q)=s_{11}(p)=0$. Using Lemma~\ref{lem:nosing}
this implies that
$\frac{\partial P}{\partial z}(q)=\frac{\partial Q}{\partial z}(q)=0$ such that
the tangent vector of $\CC$ at $q$ is vertical. Reciprocally, if the tangent
vector of $\CC$ at $q$ is vertical, then
$\frac{\partial P}{\partial z}(q)=\frac{\partial Q}{\partial z}(q)=0$ and
Lemma~\ref{lem:nosing} implies that $\frac{\partial S_2}{\partial z}(q) = 0$, thus $S_2$ has a double root in $z$ and $\det(A(p))=0$. Thus the Hessian of
$f$ is zero at $p$ and $p$ is an ordinary cusp of $f$.
\end{proof}

\subsection{Checking the assumptions}
%%%%%%%%%%%%%%%%%%%%%%%%%%%
\label{sec:checkgeneric}
As opposed to symbolic methods, our numerical approach requires assumptions on
the input. To be complete we provide a way to check that the assumptions are
fulfilled using only numerical methods.
% Following the results of Theorems~\ref{th:sing-sres} and
% \ref{th:nodecusp}, Algorithm~\ref{algo:subdiv-sing} isolates the
% singularities via the subresultant system $\{s_{11}, s_{10}\}$.
%In addition to
%assumptions $(A1)$, $(A2)$ and $(A3)$, this leads to introduce the regularity
%assumption:

% \MP{To keep the isolation algo a real algo and not a semi-algo, we may add the
%   assumptions
%   \begin{itemize}
%   \item[$(A4)$] The singularities of a resultant curve are only nodes, and the
%     singularities of a discriminant curve are only nodes or ordinary cusps. Note
%     that under assumptions $(A1)$, $(A2)$ and $(A3)$, Theorem~\ref{th:sing-sres}
%     states that it is equivalent to say that the singularities are the regular
%     solutions of the system $\{s_{11},s_{10}\}$, hence this can be checked with
%     $0\not\in \Box Jacobian(s_{11},s_{10})(B)$.
%   % \item[$(A4')$] for a resultant curve, there are only nodes. This is ensured by
%   %   adding the condition $0\not\in \Box\det(\text{Hessian}(f))(B)$ in the
%   %   Semi-algorithm~\ref{algo:assuption-check}.
%   % \item[$(A5')$] for a discriminant curve, there are only nodes or ordinary
%   %   cusps. This is ensured by adding the condition $0\not\in
%   %   \Box\det(\text{Hessian}(f))(B)$ or { $K_{(P,P_z,P_{zz})}(B\times I_z)
%   %     \subset int(B\times I_z)$ } in the
%   %   Semi-algorithm~\ref{algo:assuption-check}.
%   \end{itemize}
% } 

\begin{Lemma}
  The semi-algorithm~\ref{algo:assuption-check} terminates iff the assumptions
  $(A_1)$, $(A_2)$, $(A_3)$ and  $(A_4)$ are satisfied.
\end{Lemma}
\begin{proof}
  We first show that if the semi-algorithm terminates then $(A_1)$, $(A_2)$,
  $(A_3)$ and $(A_4)$ are satisfied.  Indeed, for any box of the subdivision,
  (a) Lines \ref{alg5} %and \ref{alg6}
  ensures that the leadings of $P$ and $Q$ have no common solutions $(A_3)$; (b)
  Lines \ref{alg7}, \ref{alg9} and \ref{alg17} ensures that $f,s_{11}$ and
  $s_{22}$ do not vanish simultaneously, hence there is at most two points of
  the curve $\CC$ above each point of $B_0$, $(A_2)$ is satisfied; (c) Lines
  \ref{alg13} and \ref{alg23} ensures that the curve $\CC$ is smooth $(A_1)$;
  Line~\ref{alg:lineA4} finally ensures the regularity assumption $(A_4)$.

  Conversely, it is easy to see that when the assumptions $(A_1)$, $(A_2)$,
  $(A_3)$ and $(A_4)$ are satisfied Semi-algorithm~\ref{algo:assuption-check}
  will terminate due to the convergence of the interval functions to the actual
  value of the corresponding function when the diameter of a box tends to 0.
\end{proof}

 % check genericity: over a $xy$-box, the isolation algo for singularities will
 %  terminate iff generic
 %  \begin{itemize}
 %  \item if $res\neq 0$ then no curve above (the 3D curve is smooth!)
 %  \item elif $sres_{11}\neq 0$, the 3d curve is smooth if in $z$ the interval
 %    $-s_{10}/s_{11}$, one of the 3 minors of the jacobian of $P,Q$ wrt $x,y,z$
 %    or $P$ or $Q$ do not vanish. Also check the leadings $L_P, L_Q$ do not both
 %    vanish.
 %  \item elif $sres_{22}\neq 0$, the one/two $z$-intervals are the interval
 %    roots of $Sres_2$, check as before that the 3d curve is smooth:
 %    non-vanishing of one of the 3 minors of the jacobian of $P,Q$ wrt $x,y,z$ or
 %    $P$ or $Q$. Also check the leadings $L_P, L_Q$ do not both vanish. 
 %    \item else subdivide
 %  \end{itemize}
\floatname{algorithm}{Semi-algorithm}
\newcommand{\continue}{\textbf{continue}}
\renewcommand{\And}{\textbf{and}\xspace}
\newcommand{\Or}{\textbf{or}\xspace}
\begin{algorithm}[t]
    \caption{Subdivision based checking of assumptions $(A_1)$, $(A_2)$, $(A_3)$ and
    $(A_4)$}
\label{algo:assuption-check}
\begin{algorithmic}[1]

  \Require{A box $B_0$ in $\R^2$ and two polynomials $P$ and $Q$ in
    $\Q[x,y,z]$.}

    \Ensure{The semi-algorithm terminates iff the assumptions $(A_1)$, $(A_2)$, $(A_3)$ and  $(A_4)$ are satisfied.}
 
  \State Let $f$ be the resultant and $s_{22},s_{11},s_{10}$ be the subresultant
  coefficients of $P$ and $Q$ wrt $z$.
% \GM{J'avais mis let pour laisser l'option
%   de garder la représentation matricielle évaluée en chaque interval}

  \State $L:=\{B_0\}$

  \Repeat
  \State $B:=L.pop$

  \If{ $0 \in \Box L_P(B)$ \And $0 \in \Box L_Q(B)$}  \label{alg5}
  \Comment{Checking $(A_3)$}
  \State {Subdivide $B$ and insert its children in $L$, \continue}\label{alg6}

  \ElsIf{$0 \not\in \Box f(B)$ }\label{alg7}
  \Comment{{\small Checking if $P$ and $Q$ have no common solution} ($A_2$)}
  \State \continue \label{alg8}

  \ElsIf{$0 \not\in \Box s_{11}(B)$}\label{alg9}
  \Comment{{\small Checking if $P$ and $Q$ have at most $1$ common solution} ($A_2$)}
  \State {$I_z:=-\Box s_{10}(B)/\Box s_{11}(B)$}
  \If{%$0 \in \Box P(B\times I_z)$ \And $0 \in \Box Q(B\times I_z)$ \And
      $(0,0,0) \in \Box \threevec{t}(B\times I_z)$}
      \label{alg13}
      \Comment{Checking $(A_1)$}
  \State {Subdivide $B$ and insert its children in $L$, \continue}
  \Else \State \continue  
  \EndIf
  \ElsIf{$0 \not\in \Box s_{22}(B)$}\label{alg17}
  \Comment{{\small Checking if $P$ and $Q$ have at most $2$ common solutions} ($A_2$)}
  \State {$I_z :=$ union of the complex boxes solution of: $\Box s_{22}(B)z^2 + \Box s_{21}(B)z + \Box s_{20}(B) = 0$}
  \If{%$0 \in \Box P(B\times I_z)$ \And $0 \in \Box Q(B\times I_z)$\\\quad \quad \And 
      $(0,0,0) \in \Box \threevec{t}(B\times I_z)$}
      \label{alg23}
  \Comment{Checking $(A_1)$}
  \State {Subdivide $B$ and insert its children in $L$, \continue}

  \ElsIf {$0\in \Box Jacobian(s_{11},s_{10})(B)$ 
    \label{alg:lineA4}
 %    \MPc{this test ensures that the isolation algo~\ref{algo:subdiv-sing}
%       terminates, an alternative could be $0\in \Box \det(\text{Hessian}(f))(B)$
%       \And in the discriminant case
%       $K_{(P,P_z,P_{zz})}(B\times I_z) \not\subset int(B\times I_z)$}
}
  \Comment{Checking $(A_4)$ }
  \State {Subdivide $B$ and insert its children in $L$, \continue}
  \Else \State \continue %\MP{+ check A4 for res or A5 for disc}
  %\IF{($0 \not\in \Box P(B\times I_{z_1})$ \OR $0 \not\in \Box Q(B\times
  %  I_{z_1})$) \AND ($0 \not\in \Box P(B\times I_{z_2})$ \OR $0 \not\in \Box Q(B\times
  %  I_{z_2})$)}
  %\State \continue
  %\ELSIF{($0 \not\in \Box P(B\times I_{z_1})$ \OR $0 \not\in \Box Q(B\times
  %  I_{z_1})$) \OR ($0 \not\in \Box P(B\times I_{z_2})$ \OR $0 \not\in \Box Q(B\times
  %  I_{z_2})$)}
  %\State {Subdivide $B$ and insert its children in $L$, \continue}
  %\ELSIF{ $\threevec{0} \in \Box \threevec{t}(B\times I_{z_1})$ 
  %  \OR $\threevec{0} \in \Box \threevec{t}(B\times I_{z_2})$}\label{alg23}
  %\State {Subdivide $B$ and insert its children in $L$, \continue}
  \EndIf
  \Else
  \State {Subdivide $B$ and insert its children in $L$, \continue}
  \EndIf
  \Until{$L=\emptyset$}

  \State \Return \textbf{true} %Assumptions $(A_1)$, $(A_2)$ and $(A_3)$ are satisfied.
\end{algorithmic}
\end{algorithm}
\floatname{algorithm}{Algorithm}

\subsection{Numerical certified isolation}

There is no new result in this section, but for the reader's convenience, we
recall a classical numerical method to isolate regular solutions of a square system
within a given domain via recursive subdivision and show how it applies in our
case. Such a subdivision method is often called branch and bound method
\cite{kearfott96} and uses the Krawczyk operator or Kantorovich theorem to
certify existence and unicity of solutions.  We recall the properties of the
Krawczyk operator and propose the naive Algorithm~\ref{algo:subdiv-sing} for the
isolation of the singularities of a resultant using the characterization of
these points proved in Theorem~\ref{th:nodecusp}.
Note that even if the assumptions $(A_1)$ to $(A_4)$ are satisfied, this naive
algorithm may fail if a singularity lies on (or near) the boundary of a box
during the subdivision.  Indeed, for this algorithm to be certified, there is a
need to use $\varepsilon$-inflation of a box when using the Krawczyk test and
cluster neighboring boxes of the subdivision. For simplicity we do not detail
this issue and refer for instance to \cite[\S 5.9]{stahl1995},\cite{kearfott1997empirical,Neumaier2005}.

%\MP{je propose de virer ce qui suit, car ca n'apporte pas grand chose et les
%  tests montrent que c'est pas efficace.
%
%Remark also that this algorithm does not necessary need to be run from the
%global input box $B_0$. The idea is that it may be more efficient in practice to
%avoid the global subdivision from a large input box.  An heuristic alternative
%is to first compute numerical approximations of the singular points with a
%non-certified algorithm, for instance using homotopy.  Algorithm
%\ref{algo:subdiv-sing} can then be run on boxes enclosing these
%approximations. Finally, the certification may be recovered globally on the box
%$B_0$ if a post-processing is able to check that there is no solution in the
%complement of the boxes certified to contain a unique singularity.
%}

Let $F$ be a mapping from $\R^2$ to $\R^2$ and denote $J_F$ its Jacobian matrix.
The following lemma is a classical tool to certify existence and uniqueness of
regular solutions of the system $F=(0,0)$. For simplicity, we state the
following lemma on $\R^2$ but this result holds in any dimension.

\begin{Lemma}(Krawczyk \cite{1969Krawczyk}\cite[\S
  7]{Rump1983}) %\cite{Nbook90}
    \label{cri:regsol}
    Let $B$ be a box in $\mathbb{R}^2$, $(x_0,y_0)$ the center
    point of $B$ and $\Delta B = \left(\begin{smallmatrix} B_x-x_0\\B_y-y_0
    \end{smallmatrix}\right)$.  Let $N$ be the mapping:
      $$ N(x,y) = \left(\begin{smallmatrix} x\\ y\end{smallmatrix}\right)
                  - J_F(x_0,y_0)^{-1}\cdot F(x,y)$$
    and $K_F$ the Krawczyk operator defined by:
      $$ K_F(B) := N(x_0,y_0) + \Box J_N(B)\cdot \Delta B .$$
      If $ K_F(B)$ is contained in the interior of $B$ then $F=(0,0)$ has a
      unique solution in $B$.
\end{Lemma}

% Je propose de ne pas commenter l'algo classique de branch and bound: c'est en
% gros expliqué dans \cite{Neumaier2005} que avec Krawczyk qui cv
% quadratiquement et une evaluation à l'ordre 2 on va terminer.
% In addition, for a small enough box enclosing a regular solution, the previous
% criterion will eventually succeed to prove the unicity \MP{ref to thesis
%   \cite[Section 5.9, thm 5.9.3]{stahl1995}, pas clair il traite la terminaison
%   dans le cas ou il n'y a pas de solution, mais rien de bien clair quand il y en
%   a ??}. 
%In the next section, we study sufficient conditions for the system
%$s_{11}=s_{10}=0$ to have regular solutions.

\begin{algorithm}[t]
  \caption{Subdivision based isolation of singularities}
\label{algo:subdiv-sing}
\begin{algorithmic}[1]

  \Require{A box $B_0$ in $\R^2$ and two polynomials $P$ and $Q$ in $\Q[x,y,z]$
    such that the assumptions $(A_1)$, $(A_2)$, $(A_3)$ and $(A_4)$ are
    satisfied.
    % \MP{+ only nodes or ordinary cusps A4 or A5?}
  }

  \Ensure{A list $L_{Sing}$ of boxes such that each box isolates a singularity
    of the curve defined by $f=Res_z(P,Q)$, and each singularity in $B_0$ is in
    a box of $L_{Sing}$.}
 
  \State Let $f$ be the resultant and $s_{22},s_{11},s_{10}$ be the subresultant
  coefficients of $P$ and $Q$ wrt $z$. 
%  \GM{on peut utiliser la forward derivation}

  %\State 
  % Let $K_{(s_{11},s_{10})}$ be the Krawczyk operator of the mapping
  % $(s_{11},s_{10}): \R^2 \longrightarrow \R^2$.

  \State $L:=\{B_0\}$

  \Repeat
  \State $B:=L.pop$

  \If{$0 \not\in \Box f(B)$ or $0 \not\in \Box s_{11}(B)$ or $0 \not\in \Box
    s_{10}(B)$} 
  \State Discard $B$
  \Else 
 {  \If{$K_{(s_{11},s_{10})}(B) \subset int(B)$ and $0 \not\in \Box s_{22}(B)$ }
    \State Insert $B$ in $L_{Sing}$
    \Else \State {Subdivide $B$ and insert its children in $L$}
    \EndIf
  }
  \EndIf
  
  \Until{$L=\emptyset$}

  \State \Return $L_{Sing}$
\end{algorithmic}
\end{algorithm}

\paragraph{Termination of Algorithm~\ref{algo:subdiv-sing}.}

We assume that $P,Q$ satisfy the assumptions $(A_1),(A_2),(A_3)$ and $(A_4)$.
Since in this case the singularities of $f$ are either nodes or ordinary
cusp points, Theorem~\ref{th:nodecusp} implies that they are
regular solutions of the system $s_{11}=s_{10}=0$. This
implies that Algorithm~\ref{algo:subdiv-sing} will always
terminate.

%\section{Local topology at singularities}
\section{Number of real branches at singularities}
%\label{sec:topology}
\label{sec:branches}

Algorithm~\ref{algo:subdiv-sing} isolates singularities in boxes. The next step
is to identify the singularity type, node or ordinary cusp, and compute the
number of real branches of the curve connected to the singular point.

% Once a singularity is isolated in a box, we want to recover the local
% topology. More precisely, given a box containing a singularity $p$, we need to:
% \begin{enumerate}
% \def\labelenumi{\arabic{enumi}.}
% \itemsep1pt\parskip0pt\parsep0pt
% \item Compute the number of real branches connected to $p$.
% \item Reduce the size of the box until it contains no closed loop.
% \end{enumerate}
% In this section we focus on counting the number of real branches connected to
% $p$.

\subsection{Resultant}\label{sec:resultant}

For a resultant curve, recall that nodes are stable singularities whereas cusps
are not, thus a purely numerical method cannot distinguish between node and cusp
singularities. In particular, given a box $B$ containing a singularity, let $I$
be a box evaluation of the determinant of the Hessian. If $I$ does not vanish in
the considered box, it is a node, but if it contains $0$, it can still be a
node, but also a cusp.
%
% In this section, we show how to recover the topology around nodes. Our
% algorithm always returns a correct answer if it terminates, and
% guarantees in this case that all the singularities are nodes.
For a node, the local topology is easily deduced from the topological
degree of the mapping $(f_x, f_y)$.

\begin{Lemma}\cite[Theorem 4.15]{alberti-topo-cagd-08}
  Let $B$ be a box containing a singularity $p$ of $f$ such that
  $I := \Box \det(H)(B)\neq 0$, then if $I<0$ then $p$ is connected to $4$
  real branches, otherwise if $I>0$, then $p$ is an isolated real point.
\end{Lemma}

Conversely, if $p$ is a node, then for a small enough box containing
$p$, the determinant of the Hessian does not contain $0$ and the number
of branches connected to $p$ can be recovered.
Thus, when $B$ contains a node singularity of the resultant,
Semi-algorithm~\ref{algo:nb-branches} will always terminate and compute the number
of real branches connected to $p$. Note that in the case when the singularity is an ordinary cusp, Semi-algorithm~\ref{algo:nb-branches} will not terminate.

\floatname{algorithm}{Semi-algorithm}
\begin{algorithm}%[t]
  \caption{Number of branches at a resultant singularity}
\label{algo:nb-branches}
\begin{algorithmic}[1]
  \Require{A box $B$ in $\R^2$ output by Algorithm~\ref{algo:subdiv-sing}
    containing a unique singular point $p$.}

  \Ensure{The number of branches connected to $p$.}
 
  \State Let $f$ be the resultant and $s_{11},s_{10}$ be the subresultant
  coefficients of $P$ and $Q$ wrt $z$. 

  % Let $K_{(s_{11},s_{10})}$ be the Krawczyk operator of the mapping
  % $(s_{11},s_{10}): \R^2 \longrightarrow \R^2$.

\While{$0\in \Box\det(\text{Hessian}(f))(B)$} 
\State $B:=B\cap K_{(s_{11},s_{10})}(B)$
\EndWhile

\If{ $\Box\det(\text{Hessian}(f))(B)>0$  }
\Return 0
\Else \,
\Return 4
\EndIf
\end{algorithmic}
\end{algorithm}
\floatname{algorithm}{Algorithm}

\subsection{Discriminant}\label{sec:discriminant}

In this section we focus on a discriminant curve.  Let $f$ be the resultant of
$P$ and $Q:=P_z$ satisfying the assumptions $(A_1),(A_2),(A_3)$ and $(A_4)$. Note that
$Res_z(P,P_z) = LT_z(P)Disc_z(P)$, assumption $(A_3)$ implies that the leading
coefficient of $P$ in $z$ is constant, such that the curve defined by $f$ is the
same as the one defined by the discriminant of $P$.

As for the resultant, the singularities of the curve $f=0$ are either nodes or ordinary cusps. Furthermore, for the discriminant curve, the ordinary cusps are stable and we can identify them numerically.
%we can detect not only nodes but also singularities $p=(\alpha,\beta)$ of higher
%multiplicity if they are the projection of a triple root of the polynomial
%$P(\alpha,\beta,z)$. Notably, in the generic case, these are the only kind of
%singularities of the discriminant.
Node singularities can be detected and their local topology computed with the
same algorithm as in the previous section for the resultant. We will now focus
on the case where the singular point is an ordinary cusp.
First we show that above an ordinary cusp, the polynomial $P$ has a triple root in $z$.

%In this
%case the determinant of the Hessian is $0$ and Lemma \ref{lem:cusp} shows
%that the singularity is an ordinary cusp.

\begin{Lemma}
    Under the assumptions $(A_1), (A_2), (A_3), (A_4)$ the point $p=(\alpha, \beta)$ is an ordinary cusp of the discriminant curve $f=0$ if and only if $P(\alpha,\beta,z)$ has a triple root in $z$.
\end{Lemma}
\begin{proof}
    Under our assumptions, Theorem~\ref{th:nodecusp} states that $p=(\alpha,\beta)$ is an ordinary cusp of the discriminant curve $f=0$ if and only if the curve $\mathcal{C}_{P\cap P_z}$ has a vertical tangent above $p$. This is the case if and only if there exists $\gamma$ such that $P_z(\alpha,\beta,\gamma) = P_{zz}(\alpha,\beta,\gamma) = 0$. Moreover, $(A_2)$ implies that $P_{zzz}(\alpha,\beta,\gamma)\neq 0$, such that $\gamma$ is a triple root of $P(\alpha,\beta,z)$.
\end{proof}

It is thus desirable to identify cusps via triple points, the following lemma
states the regularity of these points which is a necessary condition to use
numerical methods for their isolation.
%the projection of a
%triple point of $P(\alpha,\beta,z)$. 

%First we show how we can certify that a box $B$ contains the projection of a
%triple root of $P$.

%\subsubsection{Number of real branches}\label{number-of-real-branches-1}

% In the following, we focus on the case where the considered box $B$ 
% contains an ordinary cusp that is the projection of a triple point of $P$,
% and we show how to compute its local topology.

%and how to guarantee that $B$ contains no closed loop.

\begin{Lemma}
  If $P$ has a triple point, and the curve $P=P_z=0$ is smooth then the
  point is a regular solution of $P=P_z=P_{zz}=0$.
\end{Lemma}
\begin{proof}
  At the triple point $q$, the Jacobian of the system $P=P_z=P_{zz}=0$ is
  $P_{zzz}(q)\left|\begin{smallmatrix} P_x(q) & P_{xz}(q)\\ P_y(q) &
  P_{yz}(q)\end{smallmatrix}\right|$.  By assumption, $P_{zzz}(q)\neq 0$.
  Moreover, since the curve $P=P_z=0$ is regular, at least one minor of its
  jacobian matrix is not zero. Since $P_z(q)=0$ and $P_{zz}(q)=0$, this means
  that $\left|\begin{smallmatrix} P_x(q) & P_{xz}(q)\\ P_y(q) &
  P_{yz}(q)\end{smallmatrix}\right|\neq 0$.  Thus the Jacobian is not zero and
  $q$ is regular.
\end{proof}

The following more effective version of this Lemma delimits the box containing the triple root.

\begin{Lemma}[triple points]
    \label{lem:triplepoints}
    Let $B$ be a box containing a unique singular point $p$ of $f$ and assume that
  $0\notin\Box s_{22}$.
  The polynomial $P$ has a triple point in $z$ above $p$ if and only if the system $P=P_z=P_{zz}=0$ has a regular solution in the box $B\times I_z$ where $I_z$ is the interval $\frac{-\Box s_{21}}{2\Box s_{22}}$.
\end{Lemma}
\begin{proof}
  If $P(\alpha,\beta,z)$ has a triple root $z_0$ for $(\alpha,\beta)\in B$,
  then it has a multiplicity $2$ in
  $\gcd(P(\alpha,\beta,z),$ $P_z(\alpha,\beta,z))$. In particular $z_0$ is a double
  root of the second polynomial subresultant $S_2 = s_{22}z^2+s_{21}z+s_{20}$,
  and $z_0 = -\frac{s_{21}(\alpha,\beta)}{2s_{22}(\alpha,\beta)}\subset I_z$.
  Thus if $(\alpha,\beta)$ is the projection of a triple point of $P$, then this
  point is necessarily in the box $B\times I_z$.  Finally if the system $P=P_z=P_{zz}=0$ has a regular solution in $B\times I_z$, then we can
  conclude that the $3d$ box contains a triple point of $P$ and that its projection is $p$.
\end{proof}

An ordinary cusp is connected to exactly $2$ real branches. Using
Lemma~\ref{lem:triplepoints}, Algorithm~\ref{algo:nb-branches-discrim}
classifies the singularities between nodes and ordinary cusps, and compute the
number of real branches connected to them. It always terminates since the
diameter of the box converges toward $0$ such that eventually either
$\det(\text{Hessian}(f))(B)\neq 0$ or
$K_{(P,P_z,P_{zz})}(B\times I_z) \subset int(B\times I_z)$.

\begin{algorithm}%[t]
  \caption{Number of branches at a discriminant singularity}
\label{algo:nb-branches-discrim}
\begin{algorithmic}[1]
  \Require{A box $B$ in $\R^2$ output by Algorithm~\ref{algo:subdiv-sing}
    containing a unique singular point $p$.}

  \Ensure{The number of branches connected to $p$ and its singularity type (node
    or ordinary cusp).}
 
  \State Let $f$ be the resultant and $s_{2,2},s_{2,1},s_{11},s_{10}$ be the
  subresultant coefficients of $P$ and $P_z$ wrt $z$.
 % Let $K_F$ be the Krawczyk operator of the mapping $F$.
    % Let $K_{(s_{11},s_{10})}$ be the Krawczyk operator of the mapping
  % $(s_{11},s_{10}): \R^2 \longrightarrow \R^2$.

\While{true}

\If{ $\Box\det(\text{Hessian}(f))(B)>0$  }
\Return (0, node)
\EndIf
\If{ $\Box\det(\text{Hessian}(f))(B)<0$  }
\Return (4, node)
\EndIf
\State $I_z:=-\frac{\Box s_{21}(B)}{2\Box s_{22}(B)}$
\If{ $K_{(P,P_z,P_{zz})}(B\times I_z) \subset int(B\times I_z)$  }
\Return (2, ordinary cusp)
\EndIf

\State $B:=B\cap K_{(s_{11},s_{10})}(B)$
\EndWhile
\end{algorithmic}
\end{algorithm}

\section{Loop detection near singularities}
\label{sec:loopdetection}

Now that we know the number of branches $n_p$ connected to a singularity $p$, we
need to ensure that the enclosing box $B$ computed so far does not contain any
other branches not connected to $p$. First we can refine $B$ until the number of
branches crossing the boundary of $B$ matches $n_p$. But this is not enough,
since $B$ could contain closed loops of $f$. This case can be discarded by
ensuring that $B$ contains a unique solution of the system $f_x=f_y=0$.

\subsection{Resultant}

%\subsubsection{Loop detection}\label{loop-detection}

In the case of nodes, $p$ is a regular solution of the system $f_x=f_y=0$
since the determinant of the Jacobian of this system is the determinant
of the Hessian of $f$ and is not zero at $p$. Thus we can use standard
tools from interval analysis to guarantee that $p$ is the only root in $B$ of
the system $f_x=f_y=0$.

\begin{Lemma}[Node near loops]
    \label{lem:iso-node}
  Let $K_{f_x,f_y}$ be the Krawczyk operator defined in Lemma \ref{cri:regsol}
  with respect to the system $f_x=f_y=0$, and $B$ be a box containing a
  node $p$ of $f$. If $K_{f_x,f_y}(B)\subset int(B)$ then $B$ contains no
  closed loop of $f$.
\end{Lemma}

\begin{proof}
  Lemma \ref{cri:regsol} ensures that $p$ is the only solution of $f_x=f_y=0$ in
  $B$. If $B$ contains a closed loop included in $int(B)$, then a connected
  subset of $B$ has its boundary included in the curve defined by $f$. Thus it
  contains a point $q$ where $f$ reaches a local extrema and such that
  $f(q)\neq 0$. In particular, $f_x(q)=f_y(q)=0$ and $q\neq p$, hence the
  contradiction.
\end{proof}

\begin{Remark}
  \label{rem:node-loop}
  Alternatively, using tools from the next section, denoting by $\Box f$ an
  evaluation of $f$ on the box $B$, we let $I := \Box f_{xx}\Box f_{yy} -\Box
  f_{xy}\Box f_{xy} $. Then we claim that if $I$ does not contain $0$ then $B$
  contains at most 1 solution of the system $f_x=f_y=0$.
\end{Remark}

\subsection{Discriminant}
%\subsubsection{Loop detection}\label{loop-detection-1}

For the discriminant, the loops near the nodes can be handled as for the resultant. However, the same approach cannot handle ordinary cusps. The problem is that ordinary cusps are singular solutions of the system $f_x=f_y=0$. We need the following Lemma to handle ordinary cusps.
%Finally, once we know that a box contains an ordinary cusp, we want to
%ensure that it does not contain any closed loop. This is ensured by the
%following lemma.

\begin{Lemma}[Ordinary cusp near loops]
\label{lem:iso-cusp}
  Let $p$ be an ordinary cusp point of $f$ in a box $B$. Let $J, K, L, M$
  be the intervals:
% \begin{nolinenumbers}  
  \begin{dgroup*}
      \begin{dmath*}
      J = \Box f_{yy}
      \end{dmath*}
      \begin{dmath*}
      K = \Box f_{yy}^2\Box f_{xxx}-3\Box f_{yy}\Box f_{xy} 
          \Box f_{xxy}+3\Box f_{xy}^2\Box f_{xyy} -\Box f_{xy}\Box 
          f_{xx}\Box f_{yyy}
      \end{dmath*}
      \begin{dmath*}
      L = \Box f_{yy}\Box f_{xxy} + \Box f_{xx}\Box f_{yyy} - 2 \Box
          f_{xy}\Box f_{xyy}
      \end{dmath*}
      \begin{dmath*}
      M = \Box f_{yy}\Box f_{xy} - \Box f_{xy}\Box f_{yy}
      \end{dmath*}
  \end{dgroup*}
% \end{nolinenumbers}
  and let $J', K', L', M'$ be the intervals obtained by the same formula with
  $x$ and $y$ swapped. If $I = J(JK-LM)$ or $I' = J'(J'K'-L'M')$ 
  do not contain $0$, then $B$ does not contain any closed loop of the curve
  defined by $f$. \label{lem:cuspnoloop}
\end{Lemma}

\begin{Remark}
  If $B$ is small enough, then either $I$ or $I'$ does not contain zero.
\end{Remark}

When a solution of a system $S$ is singular, there are several
ways to check that a box $B$ does not contain any other solutions of $S$.
One way is to compute a univariate polynomial $r$ vanishing on the
projection of the solutions of $S$ (with resultant or Gröbner bases),
and check that the projection of $B$ contains only one solution of the
square-free part of $r$. Another way is to use a multivariate version of
the Rouch\'e theorem (\cite{Verschelde1994} for example). In our case,
this would amount to solve a system of two polynomials of degree lower
than 3 and check if these solutions are within a suitable complex box
containing $B$.

The method we propose is easy to implement and can potentially be
extended to other kinds of functions than polynomials.
The main idea behind the proof of Lemma~\ref{lem:iso-cusp} is to compute a
pseudo-resultant of $f_x$ and $f_y$ in the ring localized at $p$. Then
using the fact that the evaluation on a box of the coefficients of the
Taylor expansion of a polynomial $f$ is included in the evaluation of
the corresponding derivative of $f$, we can compute the evaluation of
the local elimination polynomial on $B$ using only derivatives of the
polynomials $f_x$ and $f_y$.

Before proving Lemma~\ref{lem:iso-cusp}, we define the notion of separation
polynomial that we will use.

%\paragraph{Separation polynomial}\label{separation-polynomial}

\begin{Definition}
  Let $S$ be a bivariate polynomial system vanishing on
  $p=(\alpha,\beta)$, and $I_S$ the ideal generated by its polynomials.
  Let $k$ be an integer and $q$ be a polynomial such that
  $q(x,y)(x-\alpha)^k \in I_S$ and $q(p)\neq 0$. Then we say that $q$ is a
  \emph{separation polynomial}.
\end{Definition}

A classical separation polynomial is obtained by computing the resultant
of $f$ and $g$ seen as univariate polynomials in $y$ with coefficients
in $K[x]$. We get a polynomial $r(x)$ that can be factorized in
$q(x)(x-\alpha)^k$ where $q(\alpha)\neq 0$. However we do not restrict $q$
to be a univariate polynomial.

\begin{Lemma}
  Let $q$ be a separation polynomial and $B$ be a box containing a solution
  $p=(\alpha,\beta)$ of $S$. If $0\notin\Box q$, then, the solutions of $S$ in
  $B$ all have the same $x$-coordinate.  Moreover, if there is a polynomial $r$
  in $I_S$ such that $0\notin\Box r_y$, then $S$ has only one solution in $B$.
\end{Lemma}

\begin{proof}
  Let $(x_0,y_0)\in B$ such that $x_0\neq \alpha$. If $q(x_0,y_0)\neq 0$, then
  $q(x_0,y_0)(x_0-\alpha)^k \neq 0$. Thus there is a polynomial in $I_S$
  that does not vanish on $(x_0,y_0)$ and this point is not a solution of
  $S$. Moreover, if $(\alpha,y_0)$ is solution of $S$ with
  $y_0\neq \beta$, then $r(\alpha,\beta)=r(\alpha,y_0)=0$ and $r_y$ has a
  solution in $B$ which contradicts the second part of the lemma.
\end{proof}

\paragraph{Proof of Lemma \ref{lem:cuspnoloop}}\label{proof-of-lemma}

Consider the system $f_x=f_y=0$. Any closed loop of $f$ contains a
solution of this system. The cusp point $p$ is also solution of this
system and if $B$ contains no other solution than $p$, then $B$ cannot
contain a loop. By hypothesis, $p$ is a cusp, hence a singular solution of
the system $f_x=f_y=0$. Thus the determinant of the
Hessian vanishes and we have: $f_{xy}(p) = f_{x^2}(p)f_{y^2}(p)$. And
since $p$ is an ordinary cusp, we know that either $f_{x^2}(p)$ or
$f_{y^2}(p)$ is not zero (otherwise the multiplicity would be $4$ or more in one
 direction). Assume without restriction of generality that
$f_{y^2}(p)\neq 0$. And let $X,Y$ be two new variables such that
$\twovector{x}{y} = M\cdot\twovector{X}{Y}$ where:
\[M = \twomatrix{f_{yy}(p)}{0}{-f_{xy}(p)}{1} \twovector{X}{Y}\]

Differentiating $f$ along the new variables, we have:
\[\twomatrix{f_{XX}}{f_{XY}}{f_{XY}}{f_{YY}} = M^T\twomatrix{f_{xx}}{f_{xy}}{f_{xy}}{f_{yy}}M\]

In particular, we have:

% \begin{nolinenumbers}
\begin{dgroup*}
\begin{dmath*}
f_{XY} = f_{yy}(p)f_{xy}-f_{xy}(p)f_{yy}
\end{dmath*}
\begin{dmath*}
f_{YY} = f_{yy}
\end{dmath*}
\begin{dmath*}
f_{XX} = f_{yy}(p)^2f_{x^2}-2f_{y^2}(p)f_{xy}(p)f_{xy}+f_{xy}(p)^2f_{yy}
       = f_{yy}(p)(f_{yy}(p)f_{xx} + f_{xx}(p)f_{yy} - 2 f_{xy}(p)f_{xy})
\end{dmath*}
\begin{dmath*}
f_{XXX}= f_{yy}(p)^3f_{xxx}-3f_{yy}(p)^2f_{xy}(p)f_{xxy}+3f_{yy}(p)f_{xy}(p)^2f_{xyy}
         -f_{xy}(p)^3f_{yyy}
       = f_{yy}(p)(f_{yy}(p)^2f_{xxx}-3f_{yy}(p)f_{xy}(p)f_{xxy}+3f_{xy}(p)^2f_{xyy}
         -f_{xy}(p)f_{xx}(p)f_{yyy})
\end{dmath*}
\end{dgroup*}
% \end{nolinenumbers}

 Observe that $f_{XY}(p) = 0$ and
$f_{XX}(p) = f_{yy}(p)(f_{xx}(p)f_{yy}(p)-f_{xy}(p)^2) = 0$. Thus, the
polynomial system $f_X,f_Y$ has the form: \[ \begin{array}{l}
f_X = a(X)\Delta X^2 + b(X,Y)\Delta Y\\
f_Y = c(X)\Delta X^2 + d(X,Y)\Delta Y\\
\end{array} \]

Eliminating $\Delta Y$, we get the polynomial $\Delta X^2(ad-c b)$ in
the ideal generated by $f_x$ and $f_y$. Letting $q = ad - cb$, we can
verify that $q(p)\neq0$. Indeed we have
$2q(p)=f_{XXX}(p)f_{YY}(p)-f_{XXY}(p)f_{XY}(p)=f_{XXX}(p)f_{YY}(p)$. By
assumption, $f_{YY}(p)=f_{yy}(p)\neq0$ and since $p$ is an ordinary
cusp, it cannot have a triple root in $X$ and $f_{XXX}(p)\neq 0$. Thus
$q$ is a separation polynomial.

Then, we can observe that $a(X) = \frac{f_X(X,\beta)}{\Delta X^2}$,
$c(X) = \frac{f_Y(X,\beta)}{\Delta X^2}$, and
$b(X,Y)=\frac{f_X-a\Delta X^2}{\Delta Y}$,
$d(X,Y) = \frac{f_Y-c\Delta X^2}{\Delta Y}$. Thus, using Taylor-Lagrange
theorem, we can deduce that if $B$ is a box containing $(\alpha,\beta)$:

\begin{align*}
    a(B) & \subset \frac{\Box f_{XXX}}{2} & c(B)&\subset \frac{\Box
    f_{XXY}}{2}\\
    b(B) & \subset \Box f_{XY} & d(B)&\subset \Box f_{YY}
\end{align*}

Finally, evaluating $2q$ on a box containing $p$, we get:

% \begin{nolinenumbers}
\begin{dmath*}
    2\Box q \subset  \Box f_{XXX}\Box f_{YY} - \Box f_{XXY}\Box f_{XY}\\
           \subset  f_{yy}(p)(f_{yy}(p)^2\Box f_{xxx}-3f_{yy}(p)f_{xy}(p)\Box
                               f_{xxy}+3f_{xy}(p)^2\Box f_{xyy} -f_{xy}(p)f_{xx}(p)\Box f_{yyy})
                              \Box f_{yy}\\
                    - f_{yy}(p)(f_{yy}(p)\Box f_{xxy} + f_{xx}(p)\Box f_{yyy} - 2 f_{xy}(p)\Box f_{xyy})
                       (f_{yy}(p)\Box f_{xy} - f_{xy}(p)\Box f_{yy})\\
           \subset  I(IJ-KL)
\end{dmath*}
% \end{nolinenumbers}

Thus if $0\notin I(IJ-KL)$ then, $0\notin \Box q$ and $0\notin \Box f_{YY}$,
thus $B$ contains no other solution of $f_x=f_y=0$ than $p$.

\subsection{Algorithm for the resultant and the discriminant curves}
\label{sec:loop}
Using the interval criteria of Lemmas \ref{lem:iso-node} and \ref{lem:iso-cusp}
for the detection of loops, Algorithm~\ref{algo:avoid-loop} returns a refined
box of a singular point that avoids closed loops of the curve, as soon as we
know in advance if the singularity is a node or an ordinary cusp. Note that this
algorithm always terminates if the singularity is a node or an ordinary cusp,
and works for any algebraic curve.

\begin{algorithm}%[t]
  \caption{Avoid curve loops in a singularity box}
\label{algo:avoid-loop}
\begin{algorithmic}[1]
  \Require{A box $B$ in $\R^2$ output by Algorithm~\ref{algo:nb-branches} or
    Algorithm~\ref{algo:nb-branches-discrim} containing a unique singular point
    $p$ with its type: node or cusp.}

  \Ensure{A box that avoids closed loops of the curve.}
 
  \State Let $f$ be the resultant and $s_{22},s_{21},s_{11},s_{10}$ be the
  subresultant coefficients of $P$ and $P_z$ wrt $z$.

  % Let $K_{(s_{11},s_{10})}$ and $K_{(f_x,f_y)}$ be the Krawczyk operators of the
  % mapping $(s_{11},s_{10})$ and $(f_x, f_y)$.

\While{true}

\If{ $B$-type = node \And  $K_{(f_x,f_y)}(B)\subset int(B)$ }
\Return $B$
\EndIf

\If{ $B$-type = cusp}
\State Compute $I$ and $I'$ as defined in Lemma~\ref{lem:iso-cusp}
\If {$0\not\in I$ \Or $0\not\in I'$ }
\Return $B$
\EndIf
\EndIf

\State $B:=B\cap K_{(s_{11},s_{10})}(B)$
\EndWhile
\end{algorithmic}
\end{algorithm}

% \section{Experiments}
% 
% \begin{itemize}
% \item Isotop for the resultant of degree 6 poly with bitsize 5: 180s
% \item Isotop for the resultant of degree 7 poly with bitsize 5: 1140s
% \item test homotopy: Hom4ps, bertini: failed; phcpack?
% \item with Alias: d6,t5: box-1..1, with 3 sol: $<$60s  sur satelite 
% \item with Alias: d7,t5: box-1..1, ?
% \item Guillaume subdivision seems longer
% \end{itemize}

\input{section_experiments.tex}

\section*{Acknowledgments}

The authors would like to thank Laurent Bus\'e and \'Eric Schost for fruitful
discussions.

\bibliographystyle{alpha}
\bibliography{bib-algcurves}

\end{document}

%% file: section_experiments.tex
\section{Experiments}
\label{sec:experiments}
% 
% \begin{itemize}
% \item Isotop for the resultant of degree 6 poly with bitsize 5: 180s
% \item Isotop for the resultant of degree 7 poly with bitsize 5: 1140s
% \item test homotopy: Hom4ps, bertini: failed; phcpack?
% \item with Alias: d6,t5: box-1..1, with 3 sol: $<$60s  sur satelite 
% \item with Alias: d7,t5: box-1..1, ?
% \item Guillaume subdivision seems longer
% \end{itemize} 

%\MP{
%  \begin{itemize}
%  \item demander à Fabrice de tester son dernier code de RS3 (dont le temps se
%  rapproche d'un calcul de resultant avec maple), avoir des binaires pour linux
%  64bit,
%\item  eliminer la colone Isolate, on peut eventuellement dire que c'est tjs
%  plus lent
%\item Pour l'homotopy: eliminer le temps par chemin, mais estimer le temps total
%  a partir de celui-ci pour Bertini en degree $>7$ en mettant une note pour
%  expliquer.
%\item Aller jusqu'au degree 9 doit etre possible
%  \end{itemize}
%}

As a proof of concept of the approach presented in this paper to compute the
topology of a singular plane curve defined by a resultant or a discriminant, we
have implemented steps (1) and (2) proposed in
Section~\ref{section-Introduction}.
Recall that step (1) consists in isolating the singularities of the curve. This
isolation is performed by Algorithm~\ref{algo:subdiv-sing} and we compare our
results with state-of-the-art symbolic and homotopic methods.
% The branch and bound algorithm~\ref{algo:subdiv-sing} handles the isolation in a
% bounded box. \cite{Nbook90} \MP{il faut une ref precise à la section} shows how
% to extend such a method in a global isolation in the plane by applying
% algorithm~\ref{algo:subdiv-sing} to three systems, each time within the initial
% box $[0,1]\times[0,1]$.
% Recall that step (1) consists in isolating singularities of the curve. It is
% partially addressed by algorithm \ref{algo:subdiv-sing} that is a branch and
% bound isolation singularities in a given box. \cite{Nbook90} \MP{il faut une ref
%   precise à la section} shows how to extend such a branch and bound method in a
% global isolation method by applying it to three systems, each time within the
% initial box $[0,1]\times[0,1]$.
In step (2), topology around singularities is computed. It is addressed in this
paper by Algorithms~\ref{algo:nb-branches} and \ref{algo:nb-branches-discrim},
that determine the number of branches at a singularity and its nature (node or
cusp), and algorithm \ref{algo:avoid-loop} that ensures that no loops lie in a
box containing a singularity.

% -----
% We implemented the above mentioned algorithms within the mathematical software
% \texttt{sage}, that provides MPFI \cite{RR04} implementation of interval
% arithmetic over a multi-precision floating arithmetic. We used the extension of
% the Krawczik operator at order two, and second order centered evaluation of
% polynomials, see \cite{Nbook90} for details.

All softwares were tested on a Intel(R) Xeon(R) CPU L5640 @ 2.27GHz machine with
Linux. Running times given here have to be understood as sequential times in seconds. 

Section~\ref{sec:soft} gives details on our implementation and the other
softwares used for comparison.  Section~\ref{subsection:isolation} presents
results of our approach for the isolation of singularities, and a comparison to
state-of-the-art symbolic and numeric methods.
Section~\ref{subsection:topology} reports our results for the computation of the
local topology at singularities.

\paragraph{Data for Tables~\ref{tab:growing_degree_constant_bitsize},
  \ref{tab:constant_degree_growing_bitsize} and \ref{tab:topology}.}

Random dense polynomials $P,Q$ are generated with given degree $d$ and bitsize
$\sigma$, that is the coefficients are integers chosen uniformly at random with
absolute values smaller than $2^\sigma$.
Unless explicitly stated, the given running times are averages over five
instances for each pair $(d,\sigma)$.

\subsection{Details of implementations}
\label{sec:soft}
\paragraph{Symbolic methods.} 
% We tested two symbolic solvers, the function \texttt{Isolate} of the
% \texttt{RootFinding} Maple package and \texttt{RSCube}, developed by Fabrice
% Rouillier available at
% \href{https://gforge.inria.fr/projects/rsdev/}{https://gforge.inria.fr/projects/rsdev/}.
% \texttt{Isolate} is a general solver based on Groebner basis and Rational
% Univariate Representations (RUR). \texttt{RSCube} is specialized for bivariate
% systems and uses triangular decompositions and RURs, it is shown in
% \cite{bouzidi:2011:inria-00580431:1,bouzidi:tel-00979707} that it is one of the
% best bivariate solvers.
% 
% The two first columns of Tables~\ref{tab:growing_degree_constant_bitsize} and
% \ref{tab:constant_degree_growing_bitsize} reports running times in seconds for
% these two methods in columns t for isolating the real solutions of the
% over-determined system of singularities $\{s_{11},s_{10},res\}$.

We tested \texttt{RS4}, developed by Fabrice Rouillier, that is specialized for bivariate
systems and uses triangular decompositions and Rational Univariate Representations(RUR); it is shown in
\cite{bouzidi:2011:inria-00580431:1,bouzidi:tel-00979707} that it is one of the
best bivariate solvers. 
Roughly speaking, it performs two steps: the first one, purely symbolic,  computes the RUR of the system. 
The second one is the numeric isolation of the solutions.
A more stable but less efficient version, called \texttt{RSCube}\footnote{available at
\href{https://gforge.inria.fr/projects/rsdev/}{https://gforge.inria.fr/projects/rsdev/}},
can be found as a package for the software \texttt{Maple}. 

The first column of Tables~\ref{tab:growing_degree_constant_bitsize} and
\ref{tab:constant_degree_growing_bitsize} reports running times in seconds for
\texttt{RS4} for isolating the real solutions of the
% over-determined system of singularities $\{s_{11},s_{10},res\}$.
system  $\{s_{11},s_{10}\}$. 
Recall that solutions of this system are singularities of the curve only if they
also are solutions of the resultant $res$.
% \RIc{Voir si la nouvelle version de RSCube accepte des contraintes}

We did also test the routine \texttt{Isolate} of the package \texttt{RootFinding} natively available within \texttt{Maple}.
Since it deals with over-determined systems, it has been used to isolate solutions of $\{s_{11},s_{10},res\}$. 
Obtained results are not reported in Tables \ref{tab:growing_degree_constant_bitsize} and
\ref{tab:constant_degree_growing_bitsize} because they are outperformed by \texttt{RS4} in every cases.

%but are briefly commented in subsection \ref{subsection:isolation}.

\paragraph{Homotopy methods.} 
We tested two homotopy solvers, \texttt{HOM4PS} \cite{HOM4PS2} and
\texttt{Bertini}\footnote{\href{https://bertini.nd.edu/}{https://bertini.nd.edu/}}. These
methods do not accept constraints, thus the isolation of the system
$\{s_{11},s_{10}\}$ is performed.
% and this system contains additional points that are not singularities of the curve. \RIc{Voir precedente remarque}
Note that the path tracking of these
software is not certified and solutions can be missed when the path tracker
jumps from one path to another. We measure the reliability of a resolution by
comparing the number of obtained complex solutions to the B\'ezout bound of the
system, which is the actual number of solutions since our systems are dense and
regular.  In Tables~\ref{tab:growing_degree_constant_bitsize} and
\ref{tab:constant_degree_growing_bitsize}, this measure is reported in the
column nsol/deg.
Notice that we tackled the problem of overflows that can arise when representing large integers by normalizing coefficients of input polynomials.   
% \MP{mention the normalization of the input coeff}

\paragraph{Subdivision method.}
We have implemented Algorithms~\ref{algo:subdiv-sing}, \ref{algo:nb-branches},
\ref{algo:nb-branches-discrim} and \ref{algo:avoid-loop} within the mathematical
software \texttt{sage}.  The critical sub-algorithms are the evaluation of
polynomials and the Krawczik operator.  Since the subresultant polynomials $s_{10}$
and $s_{11}$ have a large number of monomials with very large coefficients, an
important issue lies in both efficiency and sharpness of their interval
evaluation. We used the \texttt{fast_polynomial} library
\cite{moroz:hal-00846961} that allows to compile polynomial evaluations using
Horner scheme.  The double precision interval arithmetic of the \texttt{C++ boost}
library is used for Tables~\ref{tab:growing_degree_constant_bitsize} and
\ref{tab:constant_degree_growing_bitsize}. For Table~\ref{tab:topology}, we used
the quadruple precision interval arithmetic of MPFI \cite{RR04}.
We used the centered form at order two evaluation of polynomials that requires
to compute symbolically partial derivatives up to order two of
polynomials. Precisely, for a box $B$ with center $c$,
$\Box f(B)= f(c)+J_f(c)(B-c)+\frac{1}{2}H_f(B)(B-c)^2$ where $J_f$ is the
Jacobian and $H_f$ the Hessian of $f$. This evaluation form is studied in
\cite[\S 2.4]{Nbook90}
% in the univariate case and refer to \cite{1984Cornelius}
and proved to be quadratically convergent. It happened to be
more efficient in our experiments than the classical mean value form. 
% \RIc{parler de Krawczik a l'ordre 2}
% this is a bit surprising, a guess is that it may be better if f and f' are vanishing
In the Krawczik operator, derivatives of $s_{10}$ and $s_{11}$ are evaluated at order 1.

% In addition to the classical Krawczik operator, which is the mean value form of
% the Newton operator (Lemma~\ref{cri:regsol}), we use the order two evaluation
% form, see \cite{Nbook90} \MP{precise ref???} for details.
%
% $$K_2(x)=N(x_0)+N'(x_0)(x-x_0)+\frac{1}{2} N''(x)(x-x_0)^2$$
% Note that $N'(x_0)=1-f'(x_0)^{-1}f'(x_0)=0$ so that 
% $$K_2(x)=x_0-f'(x_0)^{-1}f(x_0)-\frac{f'(x_0)^{-1}}{2}f''(x)(x-x_0)^2$$
% Notice that this Krawczik operator at order two does not ensure uniqueness of
% solution in a box.  The classical Krawczik operator is thus checked on the boxes
% as a postprocessing step.

Algorithm~\ref{algo:subdiv-sing} performs the isolation in a bounded box.  To
extend the isolation to all real solutions, we use a method introduced by
\cite[p. 210]{Nbook90} (see also \cite[\S 5.10]{stahl1995} for a two dimensional
example). By changes of variables, this method transforms the isolation problem
in $\R^2$ to three isolations in the bounded box $[-1,1]\times[-1,1]$.  The
running times of Algorithm~\ref{algo:subdiv-sing} are given for the input box
$[-1,1]\times[-1,1]$ and for the global isolation in $\R^2$.
Concerning the isolation in $[-1,1]\times[-1,1]$, the column diam of
Tables~\ref{tab:growing_degree_constant_bitsize} and
\ref{tab:constant_degree_growing_bitsize} gives the minimum value of
$\log_{10}(diam(B))$ for all boxes $B$ either discarded or inserted in the list
of results $L_{sing}$ in Algorithm~\ref{algo:subdiv-sing}, and $diam(B)$ stands
for the diameter of $B$.

\subsection{Singularities isolation:
  Tables~\ref{tab:growing_degree_constant_bitsize} and
  \ref{tab:constant_degree_growing_bitsize}}
\label{subsection:isolation}

We analyze the results obtained with different approaches to isolate
singularities of a plane curve defined by $Resultant_z(P,Q)$ $=0$.
%\MP{est-ce qu'il y a une difference quand on prend un discriminant?} pas trop,
%se voit dans la taille de la boite
%
Table~\ref{tab:growing_degree_constant_bitsize} reports results for a constant
bitsize $\sigma=8$ and a variable degree $d$ while in
Table~\ref{tab:constant_degree_growing_bitsize} the degree is a constant $d=4$
and the bitsize $\sigma$ is the variable. 

\begin{itemize}
\item For all methods, the running times increases significantly with the degree
  of the input polynomials. 
\item Only the symbolic method has a significant increase of running time with
  the bitsize of the input polynomials.
%\item \texttt{Isolate} systematically fails when $d\geq6$ (in these cases the input 
%system has a total degree $\geq650$). When $(d,\sigma)=(5,8)$ we obtain t=23.4 seconds.
\item \texttt{HOM4PS} performs computation in double precision. Notice
  that it fails to parse input polynomials with large numbers of monomials. For
  instance, for $P, Q$ of degree 8, the subresultant polynomial $s_{10}$ has
  1326 monomials. In addition, as reported by the column nsol/deg,
  \texttt{HOM4PS} fails to find all solutions. 

 \texttt{Bertini} allows to use adaptive multi-precision and this has two
  consequences. First, \texttt{Bertini} was almost always able to isolate all
  solutions, thus we did not add the column nsol/deg as for \texttt{HOM4PS}. It
  only failed once in our experiments for a pair of input polynomials of degree 7 with
  bitsize 8, where the maximum precision of 1024 bits has been reached. Note
  also that for a degree larger than 7, we only computed a subset of the
  solutions so we cannot report on this reliability measure. Second, the
  multi-precision arithmetic has a heavy cost. 

  \texttt{Bertini} is thus more reliable but also slower than \texttt{HOM4PS}.

\item The isolation by subdivision in $\R^2$ is roughly three times more
  expensive than in the bounded box $[-1,1]\times[-1,1]$.  This is consistent
  with the fact that the isolation in $\R^2$ involves three isolations of systems
  of roughly the same complexity on this bounded box.
  
\item With constant values of $(d,\sigma)$, running times of the subdivision approach have a high variance. For instance, when $(d,\sigma)=(5,4)$ running times for the isolation in $\R^2$ are, for the five instances, $(229, 4.56, 3.03, 1.67, 2.08)$. 
\item Our approach is certified and more efficient than both homotopic and symbolic tested methods when $d>6$ for all the tests we did perform.

% \MP{  Efficiency wrt symbolic method??}
\end{itemize}

\subsection{Topology around singularities} 
\label{subsection:topology}

We focus here on the computation of the topology around singularities of
resultant and discriminant curves by applying successively
Algorithms~\ref{algo:nb-branches} or \ref{algo:nb-branches-discrim}, and
\ref{algo:avoid-loop}.
% As previously polynomial $P$ used to define discriminant curves, and $P$, $Q$
% for resultant curves, are characterized by their degree $d$ and the bit-size
% $\sigma$ of their coefficients.  For each pair $(d,\sigma)$, we have generated 5
% examples for each type of curve. 

Table~\ref{tab:topology} reports the results for different degrees $d$ and
constant bitsize $\sigma=8$ input polynomials. Algorithms~\ref{algo:nb-branches}
or \ref{algo:nb-branches-discrim}, and \ref{algo:avoid-loop} are applied on all
boxes containing singularities given by our global subdivision method.  Table
\ref{tab:topology} gives, for each type of curve and each pair $(d,\sigma)$ the
minimum, median and maximum of values $\log_{10}(diam(B))$ where $B$ are the
output boxes for which the topology is computed and certified.  The large range
of sizes for local topology certified boxes is due to the diversity of the
local geometry of the curve around a singular point: a singular point may be
near to another or near to a branch of the curve not connected to it locally. The
sizes are smaller for certifying singularities of a discriminant curve since the
test involves higher degrees polynomials to be evaluated.

\bigskip
% We finally propose to appreciate the quality of the test to avoid loops near
% cusp points of discriminant curves presented in lemma
% \ref{lem:iso-cusp}. 

We finally propose to appreciate the quality of different tests presented in
this paper on an example with a cusp and a nearby loop.  Consider the polynomial
$P_{\text{cusp}}$ defined as follows
% $$P_{\text{cusp}} = (z^3+zy-x)((y-2\delta)^2+(z-2\delta)^2+x^2-(\delta)^2)$$
$$P_{\text{cusp}} = (z^3+zx-y)((x-\delta')^2+(z-1)^2+y^2)-(\delta'/3)^2$$
Its discriminant curve with respect to $z$ is schematically drawn in the left part
of Figure~\ref{fig:loopdetection}.  This curve has a cusp point near $(0,0)$ and
a loop
% (the discriminant of the sphere $(y-2\delta)^2+(z-2\delta)^2+x^2-(\delta)^2=0$
% centered in $(2\delta,0)$ of radius $\delta$)
at a distance $\delta\simeq\delta'$ of this cusp point. The radius of the loop
is approximately $\delta$.  While the value of $\delta'$ decreases, we compute
\begin{itemize}
  % \item the time $t$ spent to compute the topology of the curve in the initial
  %   box $[-1,1]\times[-1,1]$,
\item the largest diameter $\tau_K$ of a box $B$ centered at the cusp point such
  that $K_{(s_{11},s_{10})}(B)\subset B$,
\item the largest diameter $\tau_C$ of a box $B$ centered at the cusp point such
  that Algorithm~\ref{algo:nb-branches-discrim} detects that the singularity in
  $B$ is a cusp,
\item the largest diameter $\tau_L$ of a box $B$ centered at the cusp point such
  that the test of Lemma~\ref{lem:iso-cusp} is satisfied.
\end{itemize}

% the largest diameter $\tau_L$ (respectively $\tau_K$) 
% of a box $B$ centered in the cusp point such that the test of lemma \ref{lem:iso-cusp} is satisfied (resp. $K_{(s_{11},s_{10})}(B)\subset B$).  
The right part of figure \ref{fig:loopdetection} displays the values of
$\log_{10}(\frac{\tau_K}{\delta})$, $\log_{10}(\frac{\tau_C}{\delta})$,
$\log_{10}(\frac{\tau_L}{\delta})$ when $\log_{10}(\delta)$ varies in
$[-0.5,-6]$.  For instance, when $\delta'=2^{-16}\simeq1.5*10^{-5}$, we obtain
$\delta\simeq10^{-5}$, $\tau_L \simeq 3.9*10^{-9}$, $\tau_K\simeq3*10^{-11}$ and
$\tau_C\simeq1.7*10^{-21}$. In this very precise case, the isolation of the
singularities in the initial box $[-1,1]\times[-1,1]$ together with the
computation of the local topology with our certified numerical method takes
$2.94$ seconds.

Notice that once a singularity has been isolated in a box $B$ by the subdivision
process, the box $B'$ allowing to certify the nature of the singularity is
obtained by contracting $B$ with the Krawczik operator, which is known to be
quadratically convergent. In the above example, when $\delta'=2^{-16}$, three
iterations of the Krawczik operator are needed to obtain the suitable box.  As a
consequence, rather than having an incidence on the computation time, the high
gradient of $\tau_C$ with respect to $\delta$ leads to the need of a
multi-precision arithmetic to carry out the topology certification.

Finally one can remark that in this example the test to avoid loops presented in
Lemma~\ref{lem:iso-cusp} do not require to contract the box obtained by the
subdivision process to be fulfilled.

%%%%%%%%%%%%%%%%%%% TABLES%%%%%%%%%%%%%%%%%%
\begin{table}[p]
  \caption{Isolating singularities of $Resultant_Z(P,Q)=0$, with $P$ and $Q$ of
degree $d$ and coefficients of constant bitsize $\sigma=8$. The running times
are in seconds, the value diam is the minimum value of $\log_{10}(diam(B))$ for
all boxes $B$ considered in Algorithm~\ref{algo:subdiv-sing}, where $diam(B)$ is
the diameter of the box.}
\begin{center}
\begin{small}
%%%%%%%%%%%%%%%%%%%%%%%%%%%%%%%%%%%55 même tableau sans la colonne nsol/deg pour Bertini%%%%%%%%%%%%%%%%%%%%%%%%%%%%%%%%%%%%%%%%%%%%%
% \begin{tabular}{r|c|c|ccc|cc|cc|c|}
%               & \texttt{RSCube} & \texttt{Isolate} &\multicolumn{3}{|c|}{\texttt{HOM4PS}} &\multicolumn{2}{|c|}{\texttt{Bertini}}&\multicolumn{3}{|c|}{\texttt{Subdivision}}\\
%   domain      &$\mathbb{R}^2$&$\mathbb{R}^2$&\multicolumn{3}{|c|}{$\mathbb{C}^2$}&\multicolumn{2}{|c|}{$\mathbb{C}^2$}&\multicolumn{2}{|c|}{$[-1,1]\times[-1,1]$}&$\mathbb{R}^2$\\
%    $d,\sigma$ &   t   &    t &     t &   t/p  &nsol/deg  &   t       &   t/p   & t     & diam  & t \\\hline  
% $4,8$         & 0.89  & 1.19 &0.09   & 0.001  & 98.8\%   & 4.9       & 0.02    & 0.3   & -3.6  & 0.73 \\
% $5,8$         & 10.45 & 23.4 & 1.42  & 0.005  & 97.4\%   & 159.07    &  0.07   & 0.7   & -3.2  &  2.48 \\
% $6,8$         & 90    & (1)  & 17.3  & 0.03   & 87.5\%   & 13064 (3) & 15.5 (3)& 2.1   & -3.5  & 11.6  \\
% $7,8$         & 657   & (4)  &101    & 0.07   & 80.4\%   & 83120 (3) & 47  (3) & 8.1   & -3.8  & 23.2  \\
% $8,8$         & 4227  & (4)  & (2)   & (2)    & (2)      & (4)       & 156 (3) & 73.3  & -5.2  & 247.9 
% \end{tabular}
%%%%%%%%%%%%%%%%%%%%%%%%%%%%%%%%%%%p%%%%%%%%%%%%%%%%%%%%%%%%%%%%%%%%%%%%%%%%%%%%%
\begin{tabular}{r|c|cc|c|cc|c|}
              & \texttt{RS4} &\multicolumn{2}{|c|}{\texttt{HOM4PS}} &\texttt{Bertini} &\multicolumn{3}{|c|}{\texttt{Subdivision}}\\
  domain      &$\mathbb{R}^2$&\multicolumn{2}{|c|}{$\mathbb{C}^2$}&$\mathbb{C}^2$&\multicolumn{2}{|c|}{$[-1,1]\times[-1,1]$}&$\mathbb{R}^2$\\
   $d,\sigma$ &   t      &     t &nsol/deg  &   t       & t     & diam  & t \\\hline  
$4,8$         & 0.214    & 0.078  & 98.6\% & 3.256           & 0.435  & -3.2   & 1.071 \\ 
$5,8$         & 2.845    & 1.543  & 96.3\% & 124.774         & 0.682  & -3.0   & 2.678 \\ 
$6,8$         & 23.90    & 15.18  & 90.3\% & 1604 (2)        & 3.067  & -3.8   & 9.630 \\ 
$7,8$         & 137.9    & 97.95  & 75.5\% & 83120 (2)       & 8.469  & -4.4   & 27.43 \\ 
$8,8$         & 725.7    & (1)    & (1)    & 382200 (2,3)    & 43.47  & -5.0   & 82.98 \\ 
$9,8$         & 2720 (2) & (1)    & (1)    & 2766400 (2,3)   & 47.25  & -4.8   & 273.2 \\
\end{tabular}
\end{small}
\end{center}
% (1) Fails with segmentation false (groebner basis?)\MP{report on the Isolate
%   function of maple in the comments}

(1) Fails with segmentation false (does not support polynomials with large number of terms)

(2) Has been run on a unique example

(3) Time has been obtained by interpolating the time spent for tracking a unique path
\label{tab:growing_degree_constant_bitsize}
\end{table}

\begin{table}[p]
  \caption{Isolating singularities of $Resultant_Z(P,Q)=0$, with $P$ and $Q$ of
    constant degree $d=5$ and coefficients of bitsize $\sigma$. The running times
    are in seconds, the value diam is the minimum value of $\log_{10}(diam(B))$ for all
    boxes $B$ considered in Algorithm~\ref{algo:subdiv-sing}, where $diam(B)$ is
    the diameter of the box. 
  } 
\begin{center}
\begin{small}
\begin{tabular}{r|c|cc|c|cc|c|}
              & \texttt{RS4} &\multicolumn{2}{|c|}{\texttt{HOM4PS}} &\texttt{Bertini} &\multicolumn{3}{|c|}{\texttt{Subdivision}}\\
  domain      &$\mathbb{R}^2$&\multicolumn{2}{|c|}{$\mathbb{C }^2$}&$\mathbb{C}^2$&\multicolumn{2}{|c|}{$[-1,1]\times[-1,1]$}&$\mathbb{R}^2$\\
   $d,\sigma$ &   t   &     t &nsol/deg  &   t       & t       & diam  & t \\\hline 
%   $4,4 $      & 0.4   & 0.06  & 98\%     & 2.5       & 0.61    & -3.8  & 1.46 \\
%   $4,8$       & 0.89  & 0.09  & 99\%     & 4.9       & 0.3     & -3.6  & 0.73 \\
%   $4,16$      & 1.6   & 0.12  & 98\%     & 2         & 0.25    & -2.8  & 0.98 \\
%   $4,32$      & 4     & 0.09  & 99\%     & 2.5       & 0.29    & -2.6  & 0.67 \\
%   $4,64$      & 10.6  & 0.04  & 99\%     & 1.8       & 0.2     & -2.4  & 0.65 \\
  $5,4 $      & 1.788 & 1.532 & 94.63\%  & 263.4     & 0.755   & -3.2  & 48.13 \\
  $5,8$       & 2.845 & 1.543 & 96.32\%  & 124.7     & 0.682   & -3.0  & 2.678 \\
  $5,16$      & 4.687 & 1.431 & 93.60\%  & 300.2     & 7.052   & -4.2  & 19.22 \\ 
  $5,32$      & 7.468 & 1.817 & 94.48\%  & 264.2     & 2.439   & -3.6  & 7.173 \\ 
  $5,64$      & 13.33 & 1.728 & 96.98\%  & 233.7     & 1.906   & -3.4  & 4.676 \\
\end{tabular}

\end{small}
\end{center}
\label{tab:constant_degree_growing_bitsize}
\end{table}

\begin{table}[p]
  \caption{Computing topology around singularities of discriminant
    (resp. resultant) curves when $P$ (resp. $P,Q$) has degree $d$ and constant
    bit-size $\sigma=8$. The values min, med and max are the minimum, median and
    maximum values of $\log_{10}(diam(B))$ where $B$ are the output boxes for which
    the topology is computed and certified by Algorithms~\ref{algo:nb-branches} or
    \ref{algo:nb-branches-discrim}, and \ref{algo:avoid-loop} and $diam(B)$ is
    the diameter of the box. }
 \begin{center}
\begin{small}
\begin{tabular}{r|c|c|c||c|c|c|}
              &\multicolumn{3}{|c||}{Resultant} &\multicolumn{3}{|c|}{Discriminant} \\
   $d,\sigma$ &   min         & med        & max        &  min         & med        & max       \\\hline 
$4,8$         & $-5$          & $-4$       & $-2$       &$-12$         & $-5$       & $-3$ \\
$5,8$         & $-6$          & $-4$       & $-2$       &$-9$          & $-4$       & $-3$ \\ 
$6,8$         & $-8$          & $-5$       & $-3$       &$-17$         & $-5$       & $-2$ \\
$7,8$         & $-9$          & $-4$       & $-3$       &$-15$         & $-6$       & $-3$ \\
$8,8$         & $-12$         & $-5$       & $-3$       &$-12$         & $-6$       & $-2$ \\ 
$9,8$         & $-9$          & $-5$       & $-3$       &$-15$         & $-6$       & $-3$ \\
% $4,8$         & -8            &      -3    &      -2    &        -7   &       -3    &   -3      \\
% $5,8$         & -8            &      -4    &      -3    &   -9        &       -6    &   -4      \\
% $6,8$         & -10           &      -4    &      -2    &   -13       &       -5    &   -3     \\
% $7,8$         & -7            &      -4    &      -3    &   -15       &       -6    &   -4     \\
% $8,8$         & -10           &      -5    &      -3    &   -15       &       -6    &   -4     \\
\end{tabular}
\end{small}
\end{center}
\label{tab:topology}
\end{table}

\begin{figure}[p]
 \begin{minipage}{0.4\linewidth}
 \centering
  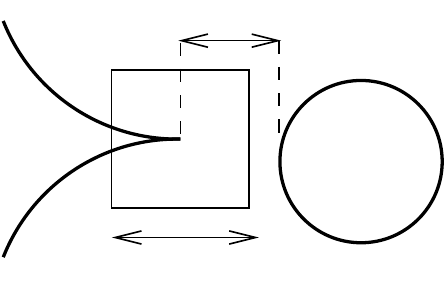
 \end{minipage}
 \begin{minipage}{0.6\linewidth}
  \centering
  \includegraphics[width=9cm]{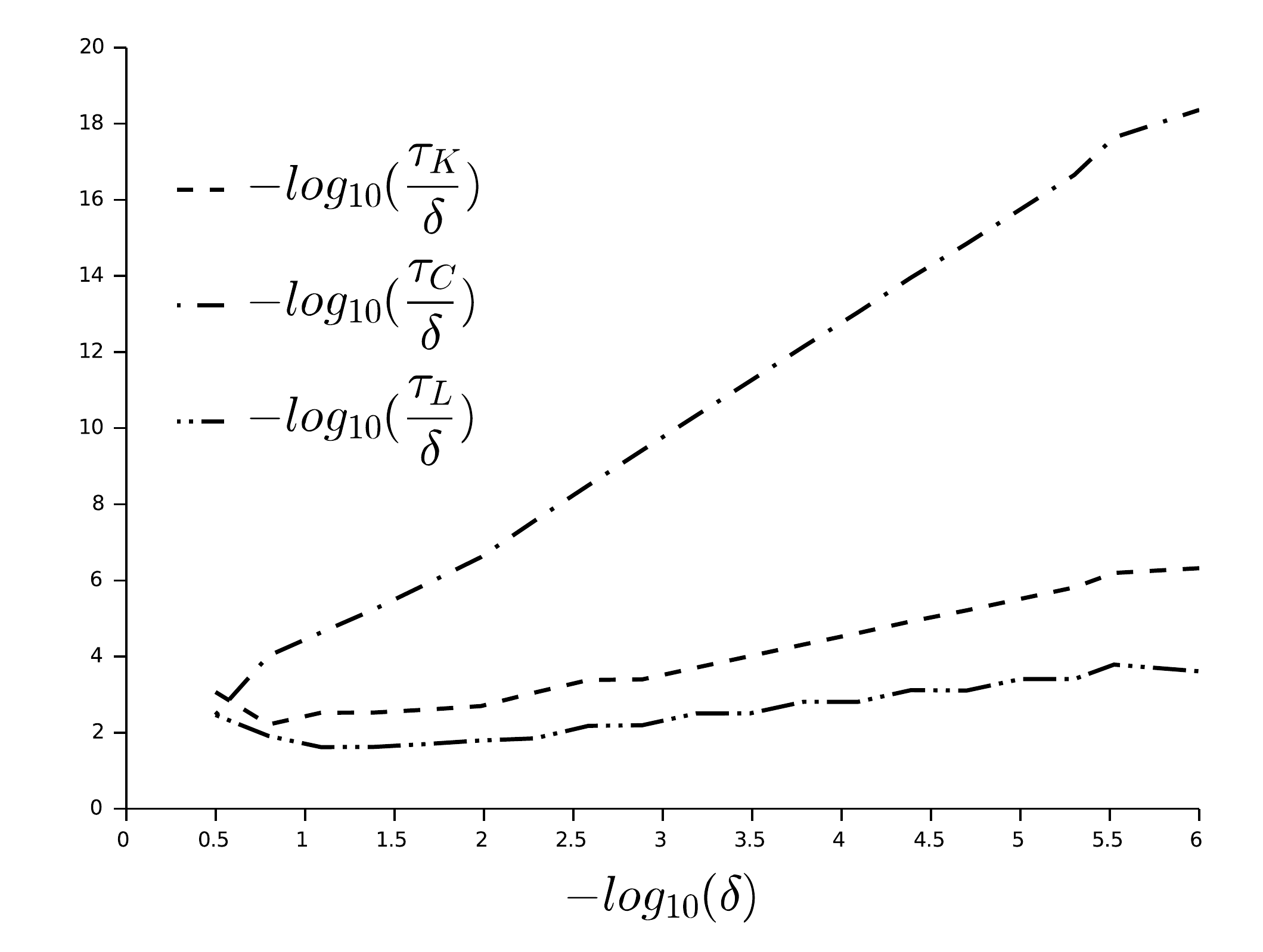}
 \end{minipage}
 \caption{Left: a schematic representation of the discriminant of the polynomial
   $P_{\text{cusp}}$. Right: largest diameters $\tau_K$, $\tau_C$, $\tau_L$ of a certified box as a function of
   the parameter $\delta$.}
 \label{fig:loopdetection}
\end{figure}

%% file: cusp_sphere_disc.pdf_t
\begin{picture}(0,0)%
\includegraphics{cusp_sphere_disc.pdf}%
\end{picture}%
\setlength{\unitlength}{4144sp}%
\begingroup\makeatletter\ifx\SetFigFont\undefined%
\gdef\SetFigFont#1#2#3#4#5{%
  \reset@font\fontsize{#1}{#2pt}%
  \fontfamily{#3}\fontseries{#4}\fontshape{#5}%
  \selectfont}%
\fi\endgroup%
\begin{picture}(2037,1331)(3676,-4359)
\put(4682,-3175){\makebox(0,0)[lb]{\smash{{\SetFigFont{12}{14.4}{\familydefault}{\mddefault}{\updefault}{\color[rgb]{0,0,0}$\delta$}%
}}}}
\put(4457,-4295){\makebox(0,0)[lb]{\smash{{\SetFigFont{12}{14.4}{\familydefault}{\mddefault}{\updefault}{\color[rgb]{0,0,0}$\tau$}%
}}}}
\end{picture}%